\documentclass[journal]{IEEEtran}

\usepackage{setspace}
\usepackage{lettrine}
\usepackage{graphics,
           psfrag,
           epsfig,
           amsthm,
           cite,
           amssymb,
           url,
           dsfont,
           subfigure,
           algorithm,
           algorithmic,
           balance,
           enumerate,
           color,
           setspace
}
\usepackage{amsmath}%[centertags][tbtags][sumlimits][nosumlimits][intlimits][namelimits][nonamelimits]

\usepackage{epstopdf}

\newtheorem{lemma}{Lemma}

\newtheorem{theorem}{Theorem}

\newcommand{\eref}[1]{(\ref{#1})}
\newcommand{\sref}[1]{Section~\ref{#1}}

\newcommand{\fref}[1]{Figure~\ref{#1}}

\newcommand{\cref}[1]{Constraint~\ref{#1}}
\newcommand{\thref}[1]{Theorem~\ref{#1}}

\newcommand{\tref}[1]{Table~\ref{#1}}
\newcommand{\ignore}[1]{}

\IEEEoverridecommandlockouts
\usepackage{mathtools}
\usepackage{color}
\newcommand*{\Perm}[2]{{}^{#1}\!P_{#2}}%

\begin{document}
\IEEEoverridecommandlockouts

\title{\vspace{-.5cm}Throughput Maximization in Cloud-Radio Access Networks using Rate-Aware Network Coding}

\author{
 \IEEEauthorblockN{Mohammed S. Al-Abiad, Ahmed Douik, \textit{Student Member, IEEE}, \\Sameh Sorour, \textit{Senior Member, IEEE} and Md Jahangir Hossain, \textit{Senior Member, IEEE}\vspace{-.8cm}}

\thanks {
A part of this paper \cite{1} is accepted in IEEE International Conference on Communications Workshops (ICCW' 2018), Kansas City, MO, USA.

Mohammed S. Al-Abiad and MD Jahangir Hossain are with the School of Engineering, University of British Columbia, Kelowna, BC V1V 1V7, Canada (e-mail: m.saif@alumni.ubc.ca, jahangir.hossain@ubc.ca).

Ahmed Douik is with the Department of Electrical Engineering, California Institute of Technology, Pasadena, CA 91125 USA (e-mail: ahmed.douik@caltech.edu).

Sameh Sorour is with the Department of Electrical and Computer Engineering, University of Idaho, Moscow, ID 83844, USA (e-mail: samehsorour@uidaho.edu)
}}

\maketitle

\begin{abstract}
One of the most promising techniques for network-wide interference management necessitates a redesign of the network architecture known as cloud radio access network (CRAN). The cloud is responsible for coordinating multiple Remote Radio Heads (RRHs) and scheduling users to their radio resources blocks (RRBs). The transmit frame of each RRH consists of several orthogonal RRBs each maintained at a certain power level (PL). While previous works considered a vanilla version in which each RRB can serve a single user, this paper proposes mixing the flows of multiple users using instantly decodable network coding (IDNC). As such, the total throughput is maximized. The joint user scheduling and power adaptation problem is solved by designing, for each RRB, a subgraph in which each vertex represents potential user-RRH associations, encoded files, transmission rates, and PLs for one specific RRB. It is shown that the original problem is equivalent to a maximum-weight clique problem over the union of all subgraphs, called herein the CRAN-IDNC graph. Extensive simulation results are provided to attest the effectiveness of the proposed solution against state of the art algorithms. In particular, the presented simulation results reveal that the method achieves substantial performance gains for all system configurations which collaborates the theoretical findings.
\end{abstract}

\begin{IEEEkeywords}
Cloud radio access networks, coordinated scheduling, instantly decodable network coding, power allocation.
\end{IEEEkeywords}

\section{Introduction} \label{sec:I}

\IEEEPARstart{O}{ver} the last decade, the continuously increasing demand for high-speed data transfer has been generating a severe burden on the wireless networks infrastructure. Moreover, the scarcity of the radio resources raises extra challenges for the Next Generation Mobile Networks 5G to meet the expected quality of service requirements \cite{Andrews_2014}. The steady move towards dense cellular architectures in 4G partially solved the problem but raised concerns regarding interference management. Cloud-radio access networks (CRANs) is one of the most promising techniques for the Next Generation Mobile Networks 5G due to their high capabilities in mitigating interference and thus providing high data rates \cite{Cai_2014}. In CRANs, a central computing unit, known as the \emph{cloud}, coordinates multiple remote radio heads (RRHs) of different sizes and possible from different tiers. Due to its centralized computation, CRANs display great ability in allocating the radio resources of the different RRHs resulting in efficient interference management.

The joint user scheduling and power adaptation in CRANs is shown to be notoriously challenging. Indeed, besides its combinatorial nature, the mathematical formulation reveals that the power adaptation is non-convex. Therefore, a large body of literature, e.g., \cite{4,5,8,18a,Park_13,park2013joint,shi2014group,14a} and references therein has been devoted to determining efficient throughput maximization schemes in different scenarios. A popular procedure to tackle the problem is to decouple the scheduling and the power adaptation problems and to iteratively solve them individually. This paper proposes both optimal and suboptimal solutions using joint and iterative approaches, respectively.

A primary limitation of the aforementioned works is their physical-layer view of the problem. Indeed, they ignore the available information in the network layer, e.g., prior file downloads. This results in assigning each resource slot to a single user, e.g., a single sub-carrier in orthogonal frequency division multiple access (OFDMA). Due to the asymmetry of the side information, coding opportunities may arise in the system which can be exploited to maximize the system throughput by mixing the flows of multiple users. Such encoding at the network layer is commonly referred to as Network Coding (NC) \cite{850663}. This paper proposes to optimally and heuristically solve the joint coordinated scheduling and power allocation problem by incorporating NC in the scheduling decisions. As such, the overall system throughput is maximized.

An important class of NC is the Opportunistic NC (ONC) \cite{19a,20a} in which the sender exploits the high diversity of previously possessed files to efficiently select the combination that would benefit a significant subset of users. Side information is widely available in modern LTE networks, due to several naturals reasons, such as erasures due to fast fading/shadowing during file transmissions, prior downloads of different popular files by different users, or protocol-based consequences, e.g., scheduling files to different users in different time epochs given their channel conditions. A particular exciting sub-class of ONC for real-time applications of interest in this paper is the instantly decodable network coding (IDNC). IDNC is a binary ONC sub-class that has engaged numerous studies, e.g., \cite{M.,A.,598452} thanks to its instant decodability properties and straightforward operations to encode/decode files using XOR-based. These simple and almost instantaneous binary operations make IDNC well-adapted for small and battery-powered devices. 

\subsection{Related Works}

As stated earlier, the coordinated scheduling problem is intrinsically combinatorial. Furthermore, determining the power levels for a pre-assigned schedule is known to be a non-convex problem. As a result, finding the optimal solution is challenging and not feasible by any polynomial-time algorithm. A large number of previous works solved the coordinated scheduling and power allocation problems separately. For example, authors of  \cite{4,5} proposed a scheduling algorithm in heterogeneous networks assuming a preassigned association of mobile users and RRHs, e.g., proportional fair scheduling.

Recent works on CRANs, e.g., \cite{Park_13,park2013joint,shi2014group}, suggested scheduling users to RRHs in a coordinated fashion by the cloud in order to maximize the network total ergodic capacity. The studies are extended in \cite{8} and references therein to include the joint optimization of the scheduling with the beamforming vectors and the power level of each radio resource. Recently, for the weighted rate maximization objective function, the authors in \cite{18a} suggested a graph theory technique to solve the joint optimization optimally.

All studies mentioned above view the network solely from the physical layer without taking into consideration upper layer facts. Therefore, each RRB serves a single user in each transmission instance. However, it has been observed that users tend to have a common interest in downloading popular files, especially videos, within a small interval of time, thus creating a pool of side information in the network. To the best of the authors' knowledge, maximizing the system throughput in a CRAN by using IDNC to allow files of different users to be scheduled simultaneously is innovative.

Different studies on IDNC revealed various code construction schemes with significant potential in minimizing multiple system parameters for different applications and network settings. For example, while the authors in \cite{M.} suggest reducing the total transmission time, i.e., the completion time, authors of \cite{S.} optimize the decoding delay. Similarly, the authors in \cite{A.} introduce a delay-based framework to reduce the completion time.

The use of IDNC in a CRAN setting brings a new set of challenges as the aim of these two techniques can be opposite. Indeed, by multiplexing their files and scheduling multiple users to the same radio resource block, the total number of targeted users increases. However, due to the heterogeneity of the achievable capacity of each user, the transmission rate of each resource block decreases. On the other hand, targeting a single user in CRAN can maximize the transmission rate of each RRB but misses the coding opportunities which could achieve further throughput and efficiency gains. Recently, IDNC had been employed in \cite{598452} in the context of a heterogeneous network setting to minimize the completion time of the users by jointly selecting the file combinations and transmission rates of each RRH. This paper aims to extend the study to the more practical and promising paradigm of CRAN. In particular, it investigates the cross-layer optimization in CRANs to achieve the maximum received throughput by jointly scheduling users to RRB/RRHs, and finding the optimal power levels of each RRB.

\subsection{Contributions}

In this work, we investigate the usage of an IDNC for throughput maximization in CRANs. Similar to the framework of \cite{4,5,8,18a}, we consider a scheduling-level coordinated CRAN in which each user can be associated to one RRH but can be served by multiple RRBs belonging to that RRH's frame. The main contributions and results of this work can be summarized as follows.

\begin{itemize}
\item Using a graph theoretical approach, we design a novel graph, called herein the CRAN-IDNC, which consists of multiple sub-graphs called power control subgraphs. Each subgraph represents the potential associations for a specific RRB wherein each vertex represents associations represented by a 6-tuple combination of RRH, RRB, user, file, transmission rate, and PL.  Such a CRAN-IDNC graph takes the following aspects into consideration:
\begin{itemize}
\item User Multiplexing: In each vertex, users are multiplexed to each RRB in each RRH based on their requests, i.e., mixing the flows of users, which is not addressed in the literature that looked at the problem from a physical-layer perspective. In this study, we consider the user multiplexing mechanism presented in \cite{598452} to deliver files to users.
\item Rate Adaptation: In order to benefit from the heterogeneity of the achievable capacity of each user, we consider the adaptive transmission rate mechanism in each vertex that represents the same RRB and the same RRH, such that each RRB selects the best transmission rate of all associated users. As such, throughput is maximized, and the QoS requirements of the end-users are maintained.
\item Optimal Power level: Power control solution is applied.  Specifically, since each vertex consists of many associations that represent scheduled users to the same RRB across all RRHs, we consider power control optimization for each vertex. This allows us to suppress the effect of interference that comes from the same RRB in different RRHs, which in turn improves the overall system throughput.
\end{itemize}
\item Using the designed CRAN-IDNC graph, the joint coordinated scheduling and power optimization problem is shown to be equivalent to a maximum-weight clique problem, which can be solved efficiently.
\item Due to the complexity of the optimal solution, we use the decoupling approach mentioned in the literature to approximate the joint optimization problem efficiently. In particular, for a fixed power, the coordinated scheduling problem is solved using similar graph theory techniques as proposed in \cite{1}. Afterward, for a fixed schedule, the power allocation problem is solved numerically. The process of iterating between coordinated scheduling and power allocation steps continues until convergence.
\item Using extensive simulations, we demonstrate the effectiveness of the proposed solution against state of the art algorithms. In particular, the presented simulation results reveal that the method achieves substantial performance gains for all system configurations which collaborates the theoretical findings.
\end{itemize}
 
The remainder of this paper is organized as follows. \sref{SMMM} introduces the considered system model and parameters. The joint scheduling and power allocation problem formulation and graph design are illustrated in \sref{C-RAN}. \sref{JSP} introduces the optimal joint solution and further designs a low-complexity heuristic. In order to further reduce the complexity, \sref{IOC} presents an iterative algorithm. Finally, before concluding in \sref{CC}, \sref{SR} plots and discusses the simulation results.

\section{System Model and Parameters} \label{SMMM}

\subsection{Cloud Radio Access Network Model}

This paper considers the downlink of a CRAN in which a computing unit is connected to a set $\mathcal{B}$ of $B$ RRHs distributed in different geographic locations within a cell and connected to the cloud through low-rate backhaul links. These $B$ RRHs serve a set $\mathcal{U}$ of $U$ mobile users in a single-hop transmission, i.e., all users are in the coverage of at least one RRH. For instance, \fref{fig1} shows a CRAN with $3$ RRHs cooperating to serve $6$ mobile users simultaneously.

The transmit frame of each RRH consists of $Z$ orthogonal time/frequency RRBs that are denoted by a set $\mathcal{Z}$ and shown in \fref{fig2}. Therefore, the total number of available RRBs in the system is $Z_{tot}=BZ$. Let $P_{bz}$ be the power allocation level (PL) of the $z$-th RRB in the $b$-th RRH. Let $\mathbf{P}=[P_{bz}]$ be a $B\times Z$ matrix containing the PLs of the considered network. From practical constraints, the power level of each RRB is bounded by $P_{bz}\leqslant P^{\max}_{bz}$. \ignore{herein for $z$-th RRB in the $b$-th RRH.} The cloud is responsible for scheduling users, synchronizing the transmission frames of all the RRHs, and determining the PL of each RRB. Due to the limited capacity of the backhaul links, each user can be assigned to at most one RRH, but possibly to many RRBs within its frame.

Let $h_{bz}^{u}(t)$ be the complex channel gain from the $z$-th RRB in the $b$-th RRH to the $u$-th user at the $t$-th transmission. The channel is assumed to be constant during the transmission time of a single uncoded/coded file and to change from one transmission to another. In our specific simulation set, we opted for the SUI model in which the channel information is affected by multiple factors, e.g., fading, shadowing, but the location of the user within the service area is the dominant factor. Such channel model leads to heterogeneous physical-layer rates from different RBBs/RRHs to different users. Moreover, we assume that there is no restriction on the model or the distribution of the channels. However, it uses the standard assumption that these values are perfectly estimated and available at the cloud. The ergodic capacity of the $u$-th user assigned to the $z$-th RRB in the $b$-th RRH can be expressed as:
\begin{equation}
R^{u}_{bz}(t)= \log_{2}(1+\text{SINR}^{u}_{bz}(\textbf{P})),
\end{equation}
wherein SINR$^{u}_{bz}(\textbf{P})$ is the corresponding signal-to-interference plus noise-ratio experienced by the $u$-th user when it is assigned with RRB $z$ of the $b$-th RRH, and can be expressed as:
\begin{equation}
\text{SINR}^{u}_{bz} (\textbf{P})= \cfrac{P_{bz} |h^{u}_{bz}|^{2}}{\sigma^2 +\sum\limits_{b^\prime\neq b}P_{b^\prime z}| h^{u}_{b^\prime z}|^{2}}, \label{e2}
\end{equation}
where $\sigma^{2}$ is the additive white Gaussian noise (AWGN) power. 

\begin{figure}[t]
\centering
\includegraphics[width=0.9\linewidth]{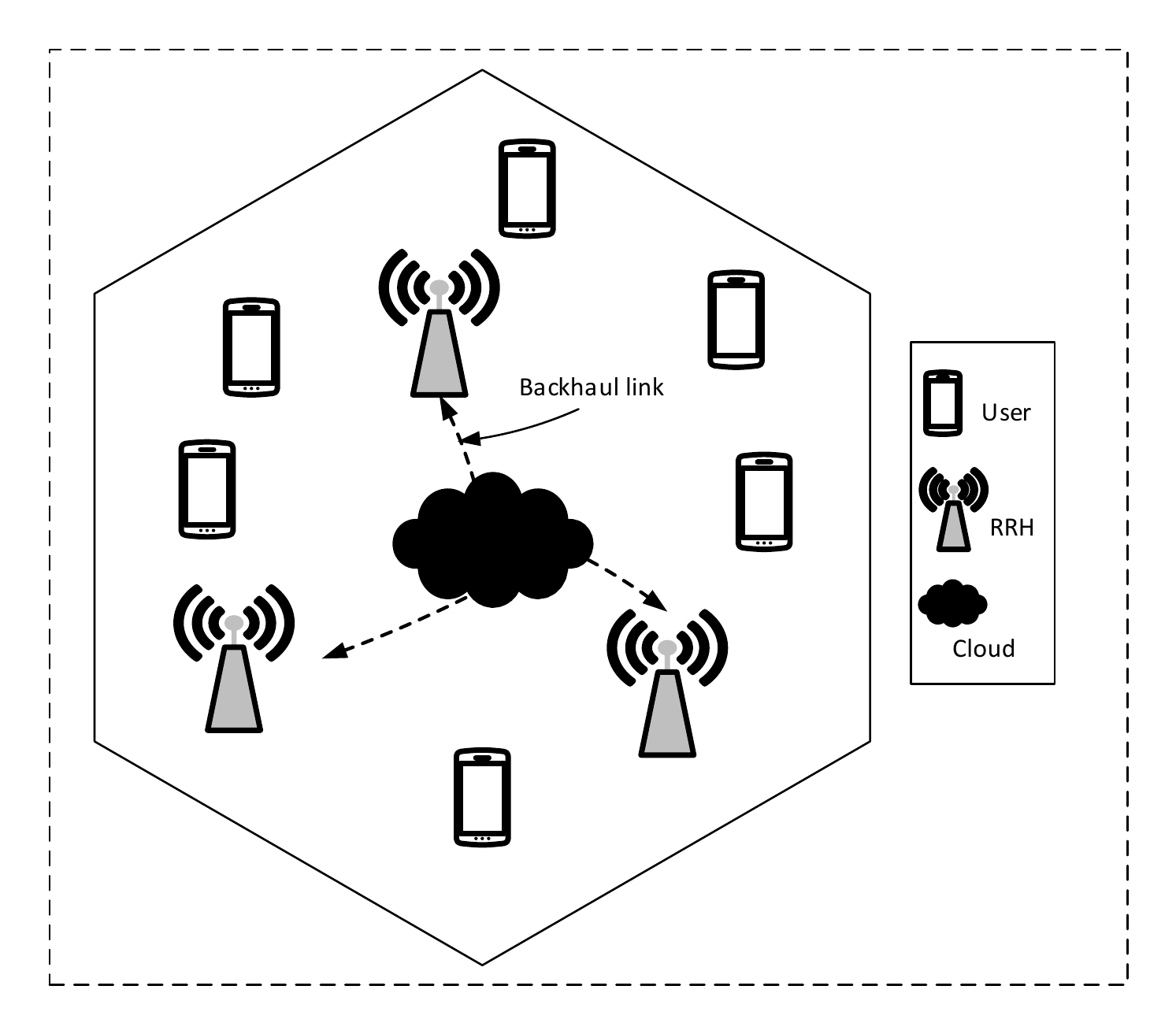}
\caption{A cloud radio access network composed of a computing unit, $6$ users, and $3$ RRHs. The cloud communicate with the RRHs through low-rate backhaul links.}
\label{fig1}
\end{figure}

\begin{figure}[t]
\centering
\includegraphics[width=0.9\linewidth]{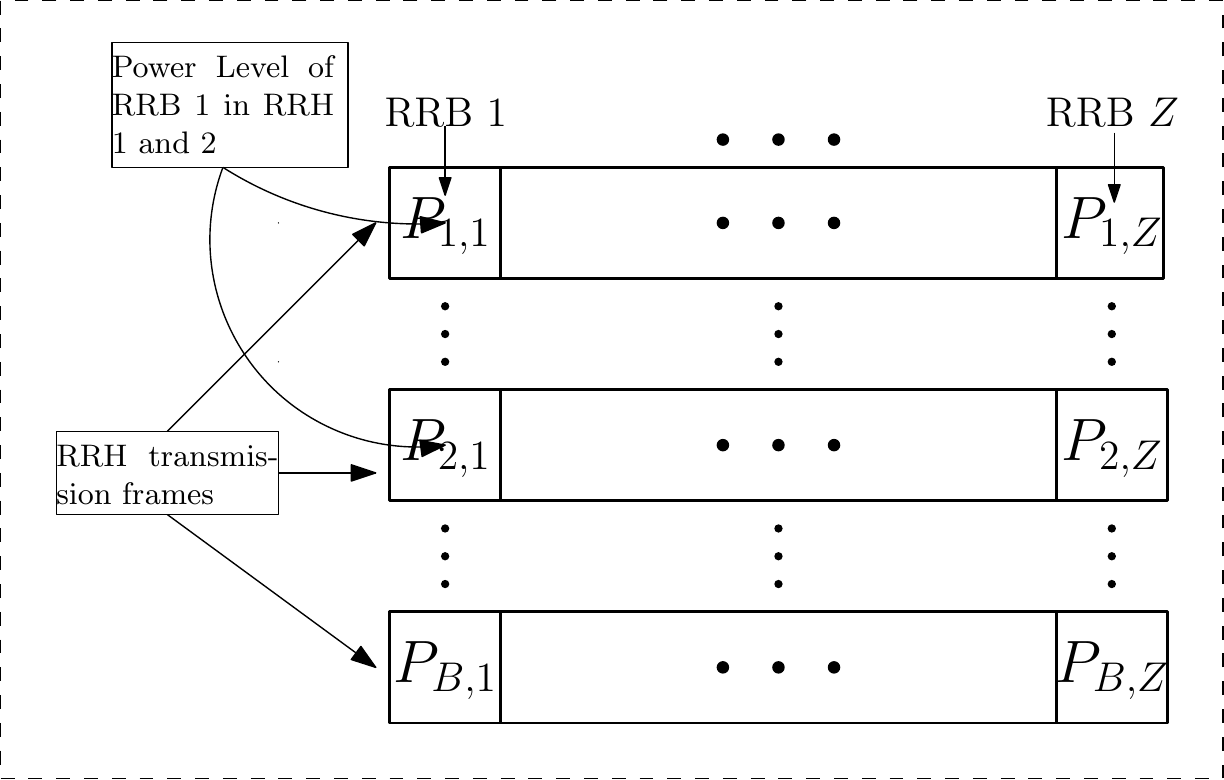}
\caption{Frame structure of the remote radio heads. Each RRH possesses $Z$ orthogonal radio resources synchronized with the RRBs of the remaining RRHs.}
\label{fig2}
\end{figure}

As stated earlier, the transmit frame of each RRH consists of $Z$ orthogonal RRBs. Therefore, interference at the $z$-th RRB is seen only from the same $z$-th RRB in the other RRHs. In other words, {SINR}$^{u}_{bz} (\textbf{P})$ depends solely on the scheduled users in $z$-th RRB across the remaining $b^\prime \neq b$ RRHs and the corresponding power level $P_{b^\prime z}$. We use in this work the standard perfect modulation assumption, i.e., the reception of an uncoded/encoded file sent in the $z$-th RRB of the $b$-th RRH is successful at the $u$-th mobile user if the transmission rate $R_{bz}$ is less than or equal the user's capacity, i.e., $R_{bz}\leqslant R_{bz}^u$. In other words, the $z$-th RRB of $b$-th RRH can transmit at a rate at most equal to the minimum ergodic capacity of its assigned users. The set of achievable capacities of all users in all RRBs across all RRHs can be represented by the set:
\begin{align}
\mathcal{R} = \bigotimes_{(b,z,u) \in \ \mathcal{\mathcal{B}} \times \mathcal{Z} \times \mathcal{\mathcal{U}}} R^u_{bz}, \label{eqw3}
\end{align} 
where the symbol $\bigotimes$ represents the product of the set of the achievable capacities. 

\subsection{Instantly Decodable Network Coding} \label{IDNC}

We assume that all users are interested in receiving/overhearing files out of a set $\mathcal{{F}}$ containing a finite library of $F$ files. These files are deemed popular due to their previous multiple downloads by different subsets of users. These popular files may represent any data format such as pictures, executable instructions, and frames from video-on-demand streaming, and thus a user can start playing the video after some (short) time for buffering, while download goes on. All files in $\mathcal{{F}}=\{f_1,f_2,\ \cdots, \,f_F\}$ are assumed to be stored in the cloud with the same size of $N$ bits so that an XOR encoding (binary operation) of any number of files, called herein encoded file, is also $N$ bits. Furthermore, the cloud keeps a log of all downloaded files by each user. Each RRH holds the whole set of files $\mathcal{{F}}$ that they receive from the cloud controller.

We assume that various RRBs in the same RRH can cooperate by sending different parts of the same uncoded/encoded file, i.e., the fragmentable constraint is allowed in this paper. Therefore, during a single transmission, users can get their requested files from the same RRH by listening to one or various RRBs. From a security perspective, the cloud controller multiplexes users to RRBs in RRHs based on their requests and encodes the communication. As such not possible for other users or attackers to decode the transmission.  

The previous users' downloaded files create an asymmetric side information in the network. Indeed, in each scheduling epoch, the files of $\mathcal{F}$ can be classified for each user $u$ as follows.
\begin{itemize}
\item The \textit{Has} set $\mathcal{H}_{u}$ containing files previously downloaded by the $u$-th user.
\item The \textit{Wants} set $\mathcal {W}_u = \mathcal{F}\backslash \mathcal {H}_u$ containing files requested by the $u$-th user in the current scheduling frame.
\end{itemize}

The cloud controller exploits such side information diversity to transmit encoded files in order to maximize the number of successfully received bits, i.e., throughput, in each scheduling frame. IDNC allows the cloud to generate XOR-encoded files using the source files in $\mathcal{F}$. Let $\tau_{bz}(\kappa_{bz})$ denotes the targeted set of users benefiting from the encoded file $\kappa_{bz}$ transmitted from the $z$-th RRB in the $b$-th RRH, where $\kappa_{bz}$ is an element of the power set $\mathcal{P(F)}$. A combination $\kappa_{bz}$ can be used to extract a new wanted file by any user $u$, i.e., instantly decodable combination, if and only if 
\begin{enumerate}
\item $R_{bz}\leqslant R_{bz}^u$: The user can properly decode the combination.
\item $| \mathcal{W}_u\cap \kappa_{bz}| =1$: The user can XOR the combination $\kappa_{bz}$ with files already downloaded to retrieve a new file.
\end{enumerate}

To illustrate the above mentioned concepts, consider the example in \fref{fig3} which illustrates a CRAN system composed of $3$ users, $2$ RRHs, and $1$ RRB per RRH frame. Each user in this example received $2$ files and missed/wants $1$ file. Assuming that the achievable capacities of all users to the RRHs is $1$ bit/second. Thanks to its coding ability, IDNC is expected to significantly outperform the uncoded schemes. Clearly, one possible solution is that RRH $1$ targets $u_1$ and $u_2$ by sending the file combination $f_1\oplus f_2$, and RRH $2$ targets $u_3$ by sending $f_3$. Thus, the file combination $\kappa_{11}=f_1\oplus f_2$ in RRH $1$ is instantly decodable for users $u_1$ and $u_2$, i.e., $u_1, u_2\in \tau_{11}(\kappa_{11})$, and the uncoded file $\kappa_{21}=f_3$ in RRH $2$ is instantly decodable only for $u_3$, i.e., $u_3\in \tau_{21}(\kappa_{21})$. The optimal achievable overall throughput in this scenario is $3$ bits/s. Indeed, the first RRH targets $u_1$ and $u_2$ with $f_1\oplus f_2$ and the second RRH serves $u_3$ with file $f_3$. Such approach definitely improves upon the $2$ bits/s throughput without coding. 

\section{Problem Formulation and Graph Construction} \label{C-RAN}

This section first formulates the joint scheduling and power adaptation optimization problem in CRANs of interest in this paper. Afterward, the section constructs CRAN-IDNC graph by designing and merging the power control subgraph of each RRB in the system. The presented concepts are illustrated using as an example presented in \fref{fig3}.

\begin{figure}[t]
\centering
\includegraphics[width=0.65\linewidth]{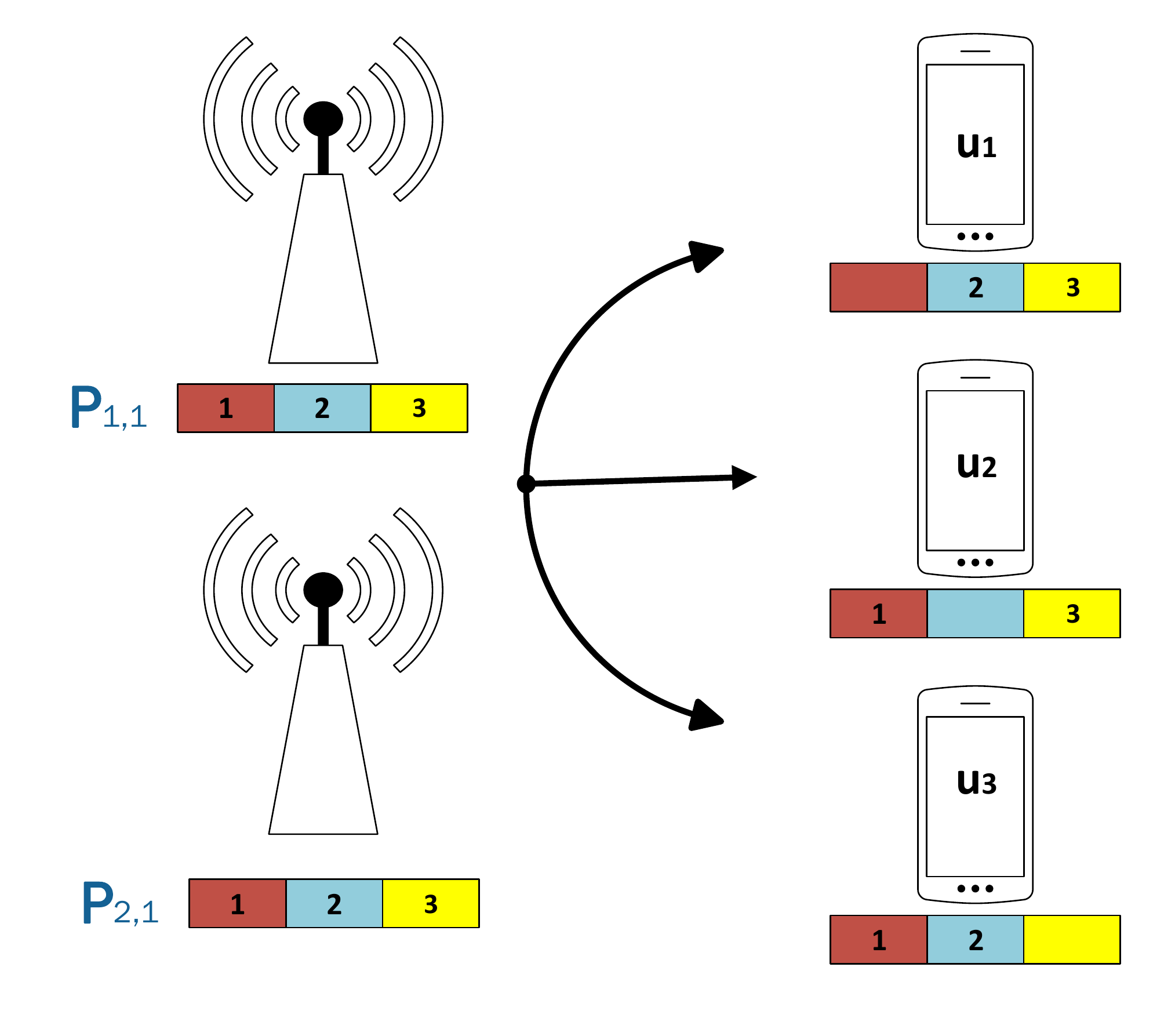}
\caption{A cloud radio access network composed of $3$ users, $3$ files, $2$ remote radio heads, and $1$ RRB in each RRH's transmit frame.}
\label{fig3}
\end{figure}

\subsection{Problem Formulation}

As stated earlier, the paper aims to improve the overall throughput in the aforementioned CRAN setting by assigning users to the RRBs of RRH and adapting the power levels under the following network connectivity constraints (CC):
\begin{itemize}
\item \textbf{CC1}: Each mobile user can connect to at most one RRH, but possibly to many RRBs in that RRH.
\item \textbf{CC2}: Each power level PL is bounded by a nominal maximal value.
\end{itemize}

Let $X_{bz}^u$ be a binary variable that is equal to $1$ if user $u$ is assigned to the $z$-th RRB of the $b$-th RRH, and zero otherwise. Let $Y_{b}^u$ be a binary variable that is set to $1$ if user $u$ is assigned to the $b$-th RRH, and zero otherwise. The joint coordinated scheduling and power allocation problem can be formulated as follows:
\begin{subequations}
\label{Cache_Scheduling_Problem}
\begin{align}
&\max\sum _{b\in \mathcal{\mathcal{B}}} \sum _{z \in \mathcal{\mathcal{Z}}} \sum _{u\in\mathcal{\tau}_{bz}(\kappa_{bz})}X_{bz}^u  \log_{2}(1+\text{SINR}^{u}_{bz}(\textbf{P}))
\label{eq4a} \\
 &{\rm s.t.\ } Y_{b}^u=\mathcal \min\left( \sum _{z}X_{bz}^u,1 \right), (b,u)\in \mathcal{B} \times \mathcal{U}   , \label{eq4b} \\
& \sum _{b}Y_{b}^u\leqslant1, u\in \mathcal{U}, \label{eq4c} \\
& \tau_{bz}(\kappa_{bz}) = \left\{u \in \mathcal{U}\ \bigg|\kappa_{bz} \cap \mathcal{W}_u = 1  \mbox{ \& } R_{bz}\leqslant R^u_{bz}\right\}, \label{eq4d}\\
& 0 \leq P_{bz} \leq P^{\max}_{bz},  (b,z)\in \mathcal{B}\times\mathcal{Z}, \label{eq4e}\\
&  X_{bz}^u , Y_{b}^u\in \{0,1\},{\kappa}_{bz}\in \mathcal{P}(\mathcal{F}),   (u,b,z)\in \mathcal{U}\times \mathcal{B}\times\mathcal{Z}, \label{eq4f}
\end{align}
\end{subequations}
where the optimization is carried over the variables $X_{bz}^u$, $Y_{b}^u$, $\kappa_{bz}$, $R_{bz}$, and $P_{bz}$. The variables $X_{bz}^u$ and $Y_{b}^u$ are discrete optimization parameters that represent the user-RRH and user-RRB associations, respectively. On the other hand, the variables $\kappa_{bz}$, $R_{bz}$ and $P_{bz}$ account for the file combination, the transmission rate, and the PLs for the $z$-th RRB of the $b$-th RRH, respectively. Constraints \eref{eq4b} and \eref{eq4c} translate \textbf{CC1}. Consequently, \eref{eq4d} ensures that all users belonging to these targeted sets $\tau_{bz}(\kappa_{bz})$ $\forall~b\in\mathcal{B}$ and $z\in\mathcal{Z}$ must receive an instantly decodable transmission. Finally, \eref{eq4e} corresponds to constraint \textbf{CC2}. 

The optimization problem \eref{eq4a} is a mixed discrete (user scheduling) and continuous (power allocation) optimization problem. Therefore, computing its global solution may need an extensive search over all possible user-to RRB/RRH associations, and determining the PL for each RRB which is not feasible for any reasonably sized network. Inspired by the work in \cite{18a}, this paper provides an efficient optimal solution to \eref{eq4a} by designing a discrete set of PLs, i.e., replacing the constraint \eref{eq4e} by the constraint $P_{bz}\in\mathbb{P}$, where $\mathbb{P}$ is a set constructed simultaneously with the CRAN-IDNC graph.

In order to construct the CRAN-IDNC graph, the rest of this subsection introduces the used notation and variables which are summarized in \tref{table_1}. Let $\mathcal{M}$ be the set of all possible associations between users and files, i.e., $\mathcal{M}=\mathcal{U}\times \mathcal{F}$. Let $\varphi_u$ and $\varphi_f$ be a family of mappings from the set $\mathcal{M}$ to the set of users $\mathcal{U}$ and the set of files $\mathcal{F}$. In other words, given the element $y=(u,f) \in \mathcal{M}$, the mappings are defined as $\varphi_u(y)=u$ and $\varphi_f(y)=f $, respectively. 

\begin{table}[!t]
\renewcommand{\arraystretch}{0.9}
\caption{Variables and Parameters of the System}
\label{table_1}
\centering
\begin{tabular}{|c|c|}
\hline
Variable & Definition\\
\hline
\hline
$\mathcal{U}$ & Set of $U$ mobile users\\
\hline
$\mathcal{F}$ & Set of $F$ files\\
\hline
$\mathcal{B}$ & Set of $B$ Remote Radio Heads (RRHs)\\
\hline
$\mathcal{Z}$ & Set of $Z$ Radio Resource Blocks (RRBs)\\
\hline
$\mathcal{R}$ & Set of all achievable capacities\\
\hline
$\mathcal{W}_u$ & Set of wanted files by user $u$\\
\hline
$R_{bz}$ & Transmission rate of RRB $z$ in RRH $b$\\
\hline
$\kappa_{bz}$ & The encoded file of RRB $z$ in RRH $b$\\
\hline
$\tau_{bz}(\kappa_{bz})$ & Set of targeted users by  $\kappa_{bz}$\\
\hline
$\mathcal{C}$ & Set of all possible IDNC file combinations\\
\hline
$\mathcal{A}$ & Set of all possible associations\\
\hline
$\mathcal{A}_{z}$ & Set of all possible associations for RRB $z$\\
\hline
\end{tabular}
\end{table}

\subsection{Power Control Subgraph Configuration}\label{RRB}

This subsection first generates all possible IDNC file combinations similar to the one proposed in \cite{S.}. The power control subgraph, then, is configured for each $z$-th RRB in the network. 

The generation of the set of the possible IDNC file combinations $\mathcal{C}$ relies on the fact that the combined associations $s$ in $\mathcal{M}$ can be served by one encoded file. Therefore, two distinct associations $s \in \mathcal{M}$ and $s^\prime\in \mathcal{M}$ are combined if one of the following conditions is satisfied:

\begin{itemize}
\item {\textbf{C1}:} $\varphi_{u}(s)= \varphi_{u}(s^\prime)$ and $\varphi_{f}(s)= \varphi_{f}(s^\prime)$. The two associations are induced by the loss of the same file $\varphi_{f}(s)$ by two distinct users $\varphi_{u}(s)$ and $\varphi_{u}(s^\prime)$.
\item {\textbf{C2}:}  $\varphi_{u}(s) \neq \varphi_{u}(s^\prime)$ and $\varphi_{f}(s)\in \mathcal{H}_{\varphi_u(s^\prime)}$ and $\varphi_{f}(s^\prime) \in \mathcal{H}_{\varphi_u(s^\prime)}$. The wanted file of each association is in the Has set of the user that induced the other association.
\end{itemize}

These conditions ensure that the IDNC file combination in $\mathcal{C}$ is always decodable for the users represented by the associations. Therefore, $\mathcal{C}$ consists of all possible IDNC file combinations, and ($\kappa,\tau$) is an element in it, i.e., $c=(\kappa,\tau$) $~\forall~ c\in \mathcal{C}$, where $\kappa$ and $\tau$ are the combination and the targeted set of users benefiting from that combination, respectively (for more details about these two components see \sref{IDNC}).

Let $\mathcal{A}$ be the set of all possible associations between RRHs, RRBs, IDNC file combinations, and achievable capacities, i.e., $\mathcal{A}=\mathcal{B}\times\mathcal{Z}\times \mathcal{C}\times\mathcal{R}$. Let $\varphi_b$, $\varphi_z$, $\varphi_c$, and $\varphi_r$ be a family of mappings from the set $\mathcal{A}$ to the set of RRHs $\mathcal{B}$, the set of RRBs $\mathcal{Z}$, the set of IDNC file combinations $\mathcal{C}$, and the set of achievable capacities $\mathcal{R}$. In other words, given the element $y=(b,z,c,R) \in \mathcal{A}$, the mappings are defined as $\varphi_b(y)=b$, $\varphi_z(y)=z$, $\varphi_c(y)=c$, i.e., $c=(\kappa,\tau$), and $\varphi_r(y)=R$, respectively. Let $\mathcal{A}_z$ represent the associations relative to the $z$-th RRB across all RRHs, i.e., $y \in \mathcal{A}_{z} \Rightarrow \varphi_z(y)=z$.

Now let the power control subgraph of the $z$-th RRB in the network be denoted by $\mathcal{G}^{z}(\mathcal{V}^{z},\mathcal{E}^{z})$ wherein $\mathcal{V}^{z}$ and $\mathcal{E}^{z}$ refer to the set of vertices and edges of this subgraph, respectively. The set of vertices is generated by merging all possible associations $s \in \mathcal{A}_{z}$ for the different RRHs under the system condition CC1, i.e., each user is associated to at most a single RRH. Therefore, each vertex $v \in \mathcal{V}^{z}$ associated with $\mathbf{S} \in \mathcal{A}_{z}$ satisfies $\varphi_{b}(s) \neq \varphi_{b}(s^\prime)$ and $\varphi_{u}(s) \neq \varphi_{u}(s^\prime)$ for all $s \neq s^\prime \in \mathbf{S}$. As a result, each vertex $v \in \mathcal{V}^{z}$ represents the partial schedule of users to the $z$-th RRB across all connected RRHs. 

The construction of the set of all vertices $\mathcal{E}^{z}$ relies on the fact that the merged \ignore{$\mathbf{S}$} associations can be served simultaneously. Therefore, each $v \in \mathcal{V}^{z}$ representing the association $\mathbf{S}$ satisfies the following conditions: 
\begin{itemize}
\item {\textbf{LC1}:} For all ($s,s^\prime)$ $ \in \mathbf{S}$ such that $\varphi_{b}(s)= \varphi_{b}(s^\prime)$, we have $\varphi_r(s)=\varphi_r(s^\prime)$. This condition guarantees that all associations in the same RRB $z$ and RRH $b$ have the same transmission rate. 
\item {\textbf{LC2}:} For all ($s,s^\prime)$ $ \in \mathbf{S}$ such that $\varphi_{b}(s) \neq \varphi_{b}(s^\prime)$, we have $\tau\cap \tau^\prime=\varnothing$. This condition guarantees that each user is scheduled to at most a single RRH.
\end{itemize}

Intuitively, assuming the power distribution $\textbf{P}$ will be computed later, a vertex $v$ representing the associations $\mathbf{S} \in \mathcal{A}_{z}$ has a weight that reflects the total contribution of the vertex to the network, i.e., the weight of $v$ can be expressed as: 
\begin{equation} 
w(v)=\sum _{s \in \mathbf{S}} \log_{2}(1+\text{SINR}^{\varphi_{u}(s)}_{\varphi_{b}(s)\varphi_{z}(s)}(\textbf{P})).
\label{eq5}
\end{equation}

\ignore{Consider the example presented in Fig. \ref{fig3} of a CRAN
system composed of $3$ users, $3$ files, $2$ RRHs, and $1$ RRB
per RRH frame. Consider the IDNC file combinations $c_1=((f_1\oplus f_2),(u_1,u_2))$ and $c_2=(f_3,u_3)$, $c_1$ and $c_2 \in \mathcal{C}$, which are belong to two possible associations $s$ and $s^\prime$ $\in \mathcal{A}_z$, i.e., $s=\{11c_1R\}$ and $s^\prime=\{21c_2R\}$. Clearly  $s$ and $s^\prime$ formed $\mathbf{S}=\{21c_1R,11c_2R\}$ as they satisfy CC1. Thus, $\mathbf{S}$ is represented by a vertex $v$ in the power control subgraph. By solving the power optimization problem for this vertex as explained in the next section, $v$ associated with $\mathbf{S}$ becomes now $v=\{1111r^*_{11}p^*_{11},1122r^*_{11}p^*_{11}, 2133r^*_{21}p^*_{21}\}$, and its weight is calculated by \eref{eq5}, i.e., $w(v)=r^*_{11}+r^*_{11}+r^*_{21}$.   }

\subsection{CRAN-IDNC Graph Design} \label{CRAN}

In \cite{598452}, the authors introduce the MB-RA-IDNC graph as a design to represent all possible file combinations, transmission rates, and users that can instantly decode the transmission for a multi base-station setting. This subsection extends the formulation to the more practical and promising paradigm of CRAN of interest in this paper. The study is further extended to include power optimization by exploiting the local power allocation graphs in a similar fashion as in \cite{18a}.

The CRAN-IDNC graph, denoted by $\mathcal{G}(\mathcal{V},\mathcal{E})$, is constructed by first generating all the $Z$ power control subgraphs. The vertex set of the CRAN-IDNC graph is simply the union of vertices of all the power control subgraphs, i.e., $\mathcal{V} = \bigcup\limits_{z \in \mathcal{Z}}\mathcal{V}^{z}$. Whereas edges between vertices within the same power control subgraph are already described in the previous subsection, the rest of this section describes remaining edges corresponding to different RRBs.

Following similar philosophy as before, two different vertices belonging to two different subgraphs are adjacent if their combination results in a feasible schedule. In particular, two vertices are connected if no user is scheduled to different RRHs. To mathematically formulate the above constraint, let vertex $v \in \mathcal{G}^{z}$ be corresponding to the association $\mathbf{S}$ and vertex $v^{\prime}\in \mathcal{G}^{z^\prime}$ corresponding to the association $\mathbf{S}^{\prime}$. Vertices $v$ and $v^{\prime}$ are adjacent if the associations they represent satisfy the following general condition (GC): 
\begin{itemize}
\item {\textbf{GC}:} For all $(s,s^{\prime})\in \mathbf{S} \times \mathbf{S}^{\prime}$ such that $\tau\cap \tau^\prime\neq\varnothing$ \ignore{$\varphi_u(s)=\varphi_u(s^{\prime})$}, we have $\varphi_b(s)=\varphi_b(s^{\prime})$. The condition insists that the same user can be scheduled only to a unique RRH.
\end{itemize}

Given the CRAN-IDNC graph $\mathcal{G}(\mathcal{V},\mathcal{E})$ as constructed above, it can be established that any maximal clique in the graph represents a set of coded transmissions that satisfies the following criterion:
\begin{itemize}
\item All users, having vertices in the maximal clique, can decode a new file from the transmission schedule of all RRBs and RRHs.
\item Each user is scheduled to a single RRH.
\item Each RRB identified by the vertices in a maximal clique adopt the transmission rate identified by the vertex. Such rate represents the smallest channel capacity of all users served by that RRB. 
\end{itemize}

Given the optimal power allocation $\textbf{P}$, the following theorem reformulates the joint coordinated scheduling and power allocation problem of interest in this paper.
\begin{theorem} \label{th:1w}
The CRAN coordinated scheduling and power allocation problem is equivalent to a maximum-weight clique problem over the CRAN-IDNC graph, and can be written as: 
\begin{align} \label{eq6}
\begin {split}
\arg \max_ {{\substack{\mathbf{C}\in \mathtt{C}} }}\sum\limits _{v\in \mathbf{C}}w(v),
\end{split}
\end{align}
where $\mathbf{C}$ is a maximal clique in the CRAN-IDNC graph, $\mathtt{C}$ is the set of all possible maximal cliques of degree $Z$, and the weight of a vertex $v \in \mathcal{V}$ is defined in \eref{eq5}. The set of targeted users and the file combination of the $z$-th RRB across all RRHs is obtained by combining the vertices of the maximum clique corresponding to the power control subgraph $\mathcal{G}^{z}$.
\end{theorem}
\begin{proof}
The proof of this theorem is omitted in this paper as it can follow the same steps used in investigating Theorem 2 in \cite{18a} but with the difference that it is applied for multiplexing users in each RRB/RRH instead of assigning one user in each RRB.
\end{proof}

\section{Joint Scheduling and Power Allocation Solution using IDNC} \label{JSP}

This section proposes a solution to the joint coordinated scheduling and power optimization problem in \eref{eq4a}. The philosophy of solution relies on solving the power control optimization problem for each vertex of the power control graph which would allow the construction of the aforementioned CRAN-IDNC graph. Afterward, using the result of \thref{th:1w}, this section shows that the optimal throughput can be reached by investigating the maximum-weight clique in the CRAN-IDNC graph, which can be solved optimally and efficiently using low-complexity graph-theoretic algorithms that available in the literature, e.g., \cite{2155446,17,16513519}.

\subsection{Optimal Scheduling and Power Control Solution}\label{OPA}

This subsection provides an efficient method for constructing a discrete set of power levels such that the optimal solution of \eref{eq4a} is reached. More specifically, we show that by simultaneously computing the optimal power levels $\textbf{P}$ while generating the CRAN-IDNC graph, we can achieve the optimal solution of the joint coordinated scheduling and power optimization problem.

Consider the $z$-th local power control graph and a vertex $v \in \mathcal{V}^z$ in that graph associated with the associations $\mathbf{S}=\{s_1,s_2,\ \cdots, \, s_{\mathtt{S}}\} \in \mathcal{A}_z$, where $\mathtt{S}$ is the degree of $\mathbf{S}$, i.e., $\mathtt{S}=|\mathbf{S}|$, and $\mathbf{S}$ represents the total targeted users of $\sum\limits_{b\in B}|\tau_{bz}(\kappa_{bz})|$. The optimal power levels PLs ($p_{1z}, \ \cdots, \, p_{Bz}$) that maximize the received throughput for that particular vertex $v$ are the solution to the following optimization problem:
\begin{align}
\label{eq7} 
&\max_{P_{bz}}\sum _{b\in \mathcal{B}} |\tau_{bz}(\kappa_{bz})| * \min_{u \in \mathcal{\tau}_{bz}(\kappa_{bz})} \log_{2}(1+\text{SINR}^{u}_{bz}(\textbf{P}))
\nonumber \\
&{\rm s.t.\ } 0\leqslant P_{bz}\leqslant P^{\max}_{bz},~\forall~ b\in \mathcal{B}, 
\end{align}
where the optimization is over the power levels $p_{bz}$, $\forall ~b\in \mathcal{B}$. 

As stated earlier, the power level $p_{bz}$ of $z$-th RRB in the $b$-th RRH depends not only on the corresponding power levels $p_{b'z}$ and on the scheduled users in these RRB. The power optimization problem \eref{eq7} is a well-known non-convex
problem \cite{38}. Despite the non-convexity of the problem, it can be solved efficiently using one of the efficient algorithms (e.g., \cite{38,39}). Our proposed solution can use any of these power optimization algorithms to solve problem \eref{eq7}. 

\ignore{For illustration, in the example of Fig. \ref{fig3}, assuming users $1$ and $2$ assigned to RRH $1$ and user $3$ assigned to RRH $2$, i.e., formed a feasible schedule $\mathbf{S}=\{1111, 1122, 2133\}$. From the fact that RRH $1$ uses the minimum achievable capacity of either user $1$ or user $2$ as a transmission rate, the power optimization problem can be solved for each user-RRB/RRH assignment, not for all users assigned to each RRB/RRH. Therefore, the power optimization problem is solved over $(1111, 2133)$ and $(1122, 2133)$, individually. }  

\subsection{Low-Complexity Greedy Algorithm}\label{Heu}
It is well established that the maximum-weight clique problem is the one of finding the clique with the maximum-weight, which is an NP-complete problem, and even its approximation is hard \cite{16}. However, it can be optimally solved with reduced complexity as compared to the naive search, e.g., the optimal algorithms in \cite{17}, \cite{16513519}\ignore{,\cite{20}}. In this work, we propose a method similar to the one in \cite{18a} to achieve the optimum of the joint coordinated scheduling and power problem \eref{eq4a} for one particular transmission.

The joint coordinated scheduling and power optimization problem can be solved by first constructing the CRAN-IDNC graph as follows. For each RRB $z\in\mathcal{Z}$, a power control subgraph $\mathcal{G}^{z}$ is generated using \textbf{LC1} and \textbf{LC2}. Afterwards, for each RRB $z\in\mathcal{Z}$ across all RRHs, a vertex $v\in\mathcal{G}^{z}$ corresponding to the feasible schedule $\mathbf{S}\in\mathcal{A}_z$ is generated for all possible associations. The optimal PLs of each association are, then, calculated by solving the optimization problem \eref{eq7}. The vertex in the power control subgraph $\mathcal{G}^{z}$ is generated by appending the computed PLs and the corresponding rates to that vertex as shown in \sref{OPA}. The same steps above are repeated for all RRBs $z\in\mathcal{Z}$. The CRAN-IDNC graph is, then, designed by merging all subgraphs and adding connections according to \textbf{GC}. The optimal solution to the joint coordinated scheduling and power optimization problem is found by solving the maximum-weight clique problem in CRAN-IDNC graph in which each iteration of finding the maximum-weight clique is implemented as follows. The algorithm computes the weight using \eref{eq5}, then the vertex with the maximum-weight $v^*$ is selected and added to $\mathbf{C}$, i.e., $\mathbf{C}$ is initially empty. The graph $\mathcal{G}$ is, then, updated by eliminating the selected vertex $v^*$ and all the vertices that are not adjacent to it so that to guarantee that the next selected vertex is not in feasible transmission conflict with the already selected ones in $\mathbf{C}$. The process continues until no
more vertices exist in $\mathcal{G}$. Clearly, the number of vertices in the selected maximum-weight clique is $Z$. The detailed procedures of the algorithm are provided in Algorithm \ref{Alg1}.
\begin{algorithm}[t]
\caption{Joint Coordinated Scheduling and Power Allocation Algorithm}
\label{Alg1}
\begin{algorithmic}[]
\REQUIRE $\mathcal{U}, \mathcal{F}, \mathcal{B}, \mathcal{Z}, \mathcal{H}_u, \mathcal{W}_u$ and $~h^u_{bz}, (u,f)\in\mathcal{U}\times\mathcal{F}$, $(b,z,c)\in\mathcal{B}\times\mathcal{Z}\times\mathcal{C}$.
\STATE Initialization: maximum-weight clique $\mathbf{C}\ = \phi$.
\FOR {$z\in\mathcal{Z}$ \do }
\STATE Initialization: $\mathcal{G}^z\ = \phi$.
\FORALL{$\mathbf{S}=\{s_1,s_2,\ \cdots, \, s_{\mathtt{S}}\}\in \mathcal{A}_z$ \do }
\STATE Solve \eref{eq8a} to compute the optimal power allocations\\ $\textbf{P}=\{(p^*_{1z},p^*_{2z},\ \cdots, \, p^*_{Bz})\}$
\STATE Create $v=\{ \{(s_1,r^*_{1z},p^*_{1z}),\cdots, \, (s_{|\mathcal{\tau}_{1z}(\kappa_{1z})|},r^*_{1z},p^*_{1z})\}$\\ $,\ \cdots, \, \{(s_{B},r^*_{Bz},p^*_{Bz}), \cdots, \,(s_{|\mathcal{\tau}_{Bz}(\kappa_{Bz})|},r^*_{Bz},p^*_{Bz})\} \}$. 
\STATE Calculate $w(v)$ using \eref{eq5}.
\STATE Set $\mathcal{G}^z= \mathcal{G}^z\cup \{v\}$. 
\ENDFOR 
\ENDFOR 
\STATE Set $\mathcal{G}= \bigcup\limits_{z\in{\mathcal{Z}}}\mathcal{G}^z$.
\STATE Connect vertices of $\mathcal{G}$ using \textbf{GC}.
\STATE Solve maximum-weight clique problem over $\mathcal{G}$. 
\STATE Output $\mathbf{C}$
\end{algorithmic}
\end{algorithm}

\begin{figure}[t]
\centering
\includegraphics[width=0.8\linewidth]{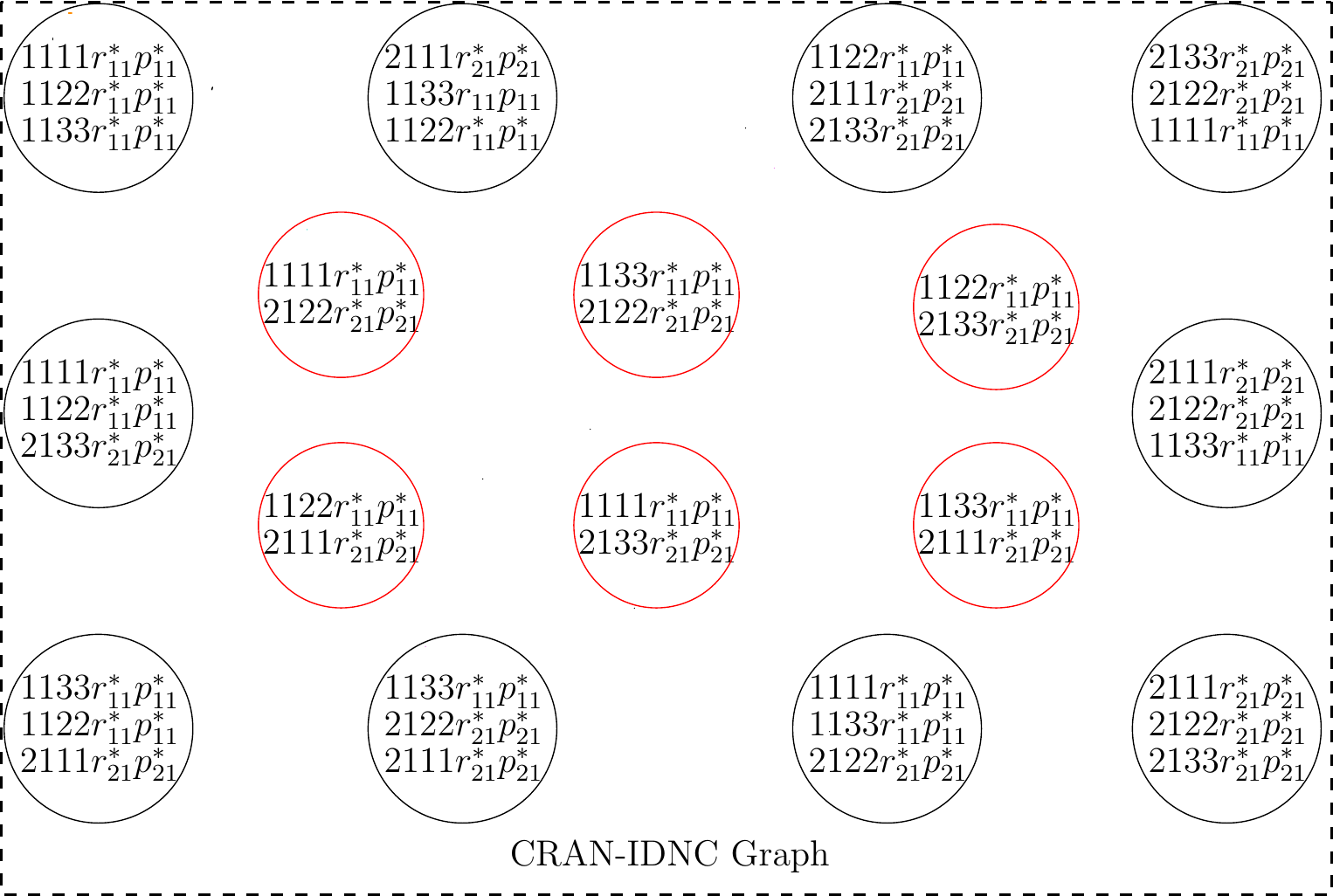}
\caption{CRAN-IDNC Graph of the example presented in \fref{fig3} based on Algorithm \ref{Alg1}.}
\label{fig4}
\end{figure}

\subsection{Motivating Example} 

Through many steps, this subsection illustrates how to use Algorithm~\ref{Alg1} to construct the CRAN-IDNC graph shown in \fref{fig4} of the example presented in \fref{fig3}. 

\textit{First step:}
In this step, we first generate the set of all possible associations $\mathcal{M}=\{u_1f_1,u_2f_2,u_3f_3\}$. Then, based on the instant decodability conditions, i.e., \textbf{C1} and \textbf{C2} that explained in \sref{RRB}, we generate all IDNC file combinations $\mathcal{C}$. \tref{table_2} summarizes all possible IDNC combinations.

\begin{table}[!t]
\renewcommand{\arraystretch}{0.9}
\caption{All Possible File Combinations $\mathcal{C}$ }
\label{table_2}
\centering
\begin{tabular}{|c|c|}
\hline
i & $c_i$($\kappa,\tau$)\\
\hline
\hline
$1$ & $((f_1\oplus f_2),(u_1,u_2))$\\
\hline
$2$ &  $((f_1),(u_1))$\\
\hline
$3$ & $((f_1\oplus f_3),(u_1,u_3))$\\
\hline
$4$ & $((f_2),(u_2))$\\
\hline
$5$ & $((f_2\oplus f_3),(u_2,u_3))$\\
\hline
$6$ & $((f_3),(u_3))$\\
\hline
$7$ & $((f_1\oplus f_2\oplus f_3),(u_1,u_2,u_3))$\\
\hline
\end{tabular}
\end{table}   

\begin{table}[!t]
\renewcommand{\arraystretch}{0.9}
\caption{All feasible Schedules $\mathbf{S}$}
\label{table_3}
\centering
\begin{tabular}{|c|c||c|c|}
        \hline
        i & $\mathbf{S}_i$ & i & $\mathbf{S}_i$\\
        \hline
        \hline
        $1$ & $\{11c_1R,21c_6R\}$ & $9$ & $\{21c_1R,11c_6R\}$\\ 
        \hline
        $2$ &  $\{11c_5R,21c_2R\}$ & $10$ & $\{21c_7R\}$\\
        \hline
        $3$ & $\{11c_7R\}$ & $11$ & $\{11c_2R,21c_4R\}$\\
        \hline
        $4$ & $\{21c_2R,11c_6R\}$ & $12$ & $\{11c_6R,21c_4R\}$\\
        \hline
        $5$ & $\{11c_6R,21c_1R\}$ & $13$ & $\{11c_4R,21c_2R\}$\\
        \hline
        $6$ & $\{11c_4R,21c_3R\}$ & $14$ & $\{21c_6R,11c_2R\}$\\
        \hline
        $7$ & $\{11c_3R,21c_4R\}$ & $15$ & $\{21c_6R,11c_4R\}$\\
        \hline
        $8$ & $\{21c_5R,11c_6R\}$ & $16$ & $\{11c_6R, 21c_2R\}$\\
        \hline
\end{tabular}
\end{table}   

\begin{table}[!t]
\renewcommand{\arraystretch}{0.95}
\caption{The represented vertices and their corresponding weights}
\label{table_4}
\centering
%\begin{tabular}{|c|c|c|}
\begin{tabular}{|p{0.2cm}| p{5.4cm}|p{1.9cm}|}
\hline
i & $v_i$ &  $w_i$\\
\hline
\hline
$1$ & $\{1111r^*_{11}p^*_{11},1122r^*_{11}p^*_{11},2133r^*_{21}p^*_{21}\}$ & $2*r^*_{11}+r^*_{21}$ \\ 
\hline
$2$ &  $\{1122r^*_{11}p^*_{11},1133r^*_{11}p^*_{11},2111r^*_{21}p^*_{21}\}$ & $2*r^*_{11}+r^*_{21}$  \\
\hline
$3$ & $\{1111r^*_{11}p^*_{11},1122r^*_{11}p^*_{11},1133r^*_{11}p^*_{11}\}$ & $3*r^*_{11}$  \\
\hline
$4$ & $\{2111r^*_{21}p^*_{21},1122r^*_{11}p^*_{11},1133r^*_{11}p^*_{11}\}$ & $r^*_{21}+2*r^*_{11}$\\
\hline
$5$ & $\{1133r^*_{11}p^*_{11},2122r^*_{21}p^*_{21},2111r^*_{21}p^*_{21}\}$ & $r^*_{11}+2*r^*_{21}$ \\
\hline
$6$ & $\{1122r^*_{11}p^*_{11},2111r^*_{21}p^*_{21},2133r^*_{21}p^*_{21}\}$ & $r^*_{11}+2*r^*_{21}$ \\
\hline
$7$ & $\{1111r^*_{11}p^*_{11},1133r^*_{11}p^*_{11},2122r^*_{21}p^*_{21}\}$ & $2*r^*_{11}+r^*_{21}$ \\
\hline
$8$ & $\{2122r^*_{21}p^*_{21},2133r^*_{21}p^*_{21},1111r^*_{11}p^*_{11}\}$ & $2*r^*_{21}+r^*_{11}$ \\
\hline
$9$ & $\{2111r^*_{21}p^*_{21},2122r^*_{21}p^*_{21},1133r^*_{11}p^*_{11}\}$ & $2*r^*_{21}+r^*_{11}$\\ 
\hline
$10$ & $\{2111r^*_{21}p^*_{21},2122r^*_{21}p^*_{21},2133r^*_{21}p^*_{21}\}$ & $3*r^*_{21}$\\
\hline
$11$ & $\{1111r^*_{11}p^*_{11},2122r^*_{21}p^*_{21}\}$ & $r^*_{11}+r^*_{21}$ \\
\hline
$12$ & $\{1133r^*_{11}p^*_{11},2122r^*_{21}p^*_{21}\}$ & $r^*_{11}+r^*_{21}$\\
\hline
$13$ & $\{1122r^*_{11}p^*_{11},2111r^*_{21}p^*_{21}\}$ & $r^*_{11}+r^*_{21}$\\
\hline
$14$ & $\{1111r^*_{11}p^*_{11},2133r^*_{21}p^*_{21}\}$ & $r^*_{11}+r^*_{21}$\\
\hline
$15$ & $\{1122r^*_{11}p^*_{11},2133r^*_{21}p^*_{21}\}$ & $r^*_{11}+r^*_{21}$\\
\hline
$16$ & $\{1133r^*_{11}p^*_{11},2111r^*_{21}p^*_{21}\}$ & $r^*_{11}+r^*_{21}$ \\
\hline
\end{tabular}
\end{table}  

\textit{Second step:}
This step generates all feasible schedules $\mathbf{S}\in \mathcal{A}_z$ such that each schedule consists of many associations $s$, in which they satisfy \textbf{CC1}, i.e., each user in each association is scheduled to one RRH. Thus, each feasible schedule $\mathbf{S}_i$ is represented by a vertex $v_i$ in the power control subgraphs based on \textbf{LC1} and \textbf{LC2} that are described in the \sref{RRB}, $\forall ~i\in |\mathbf{S}|$.  \tref{table_3} summarizes all these feasible schedules $\mathbf{S}$. In \tref{table_3}, each feasible schedule contains a set of associations each one labeled by $bzcR$, where $b$, $z$, $c$, and $R$ represent the indices of RRHs, RRBs, IDNC file combinations, and rates, respectively.

\textit{Third step:}
Given all feasible schedules $\mathbf{S}$, we solve the power optimization problem for each one, i.e., for each vertex $v$ as explained in \sref{OPA}. Having the optimal power distribution $\textbf{P}$, a vertex $v$ has a weight $w(v)$ defined in \eqref{eq5} that reflects the total contribution of the vertex to the network. Theses vertices and their weights are illustrated in \tref{table_4}.
In \tref{table_4}, it is clear that each vertex contains a set of associations each one labeled by $bzufrp$, where $b$, $z$, $u$, $f$, $r$, and $p$ represent the indices of RRHs, RRBs, users, files, achievable capacities, and PLs, respectively. 

\textit{Fourth step:}
Having all the constructed vertices and their calculation weights that are explained in the third step, we run the maximum-weight clique algorithm to find the best clique that gives the optimal solution.
Indeed, if the optimal solution that maximizes the received throughput is to schedule a single user to each RRB/RRH, then the proposed algorithm would return this schedule. Clearly, this example shows that the uncoded schemes (located in the middle of the graph with red color) are particular solutions of the proposed CRAN-IDNC algorithm, and are represented as feasible schedules by maximal cliques in the CRAN-IDNC graph. 

\subsection{Complexity Analysis of the Heuristic Solution}

It can be noted from the previous subsection that using graph theory techniques the joint scheduling and power allocation optimization problem in CRAN is equivalent to a maximum-weight clique, which is an NP-hard problem. Thus, the cloud needs, in each transmission step, to build the CRAN-IDNC graph with vertices from all power control subgraphs, and to determine its maximum-weight clique. The following lemma analysis the computational complexity of reaching the optimal solution of \eref{eq4a} using Algorithm \ref{Alg1}. 

\begin{lemma} \label{lemma1}
The complexity of finding the optimal solution of \eref{eq4a} based on Algorithm \ref{Alg1} is $\mathcal{O}(\Perm{\mathcal{C}}{B}Z(c(A)+\Perm{\mathcal{C}}{B}Z+1)+ Z)$, where $\Perm{n}{k}=\frac{n!}{(n-k)!}$ is the number of permutations, and $c(x)$ is the computational complexity of solving the power allocation problem \eref{eq7} with $x$ variables.
\end{lemma}

\begin{proof}
The complexity of the heuristic solution using Algorithm \ref{Alg1} can be decomposed in the complexity of constructing the CRAN-IDNC graph and the complexity of solving the maximum-weight clique over that graph.

We can construct the CRAN-IDNC graph by first looking at the complexity of generating a single subgraph, i.e., power control subgraph. It can be noted that the number of vertices in each subgraph is the total number of possible associations $|\mathbf{S}|$, and the number of associations relies on all possible IDNC file combinations $\mathcal{C}$. Since each subgraph is created for each $z$-th RRB across $B$ RRHs, the total number of the associations $\mathbf{S}$ in each subgraph is $\Perm{\mathcal{C}}{B}$. From \sref{RRB} it is clear that each association is represented by a vertex $v$, then, the complexity of generating the total number of vertices in each subgraph $\mathcal{O}(\Perm{\mathcal{C}}{B})$. From \sref{OPA}, we run the power optimization \eref{eq7} $\Perm{\mathcal{C}}{B}$ times to calculate the optimal power allocations for each vertex in the power control subgraph. Hence, the complexity of constructing a subgraph and solving the power optimization is $\mathcal{O}(\Perm{\mathcal{C}}{B}c(A))$. 

Before computing the maximum-weight complexity, we need first to build the adjacency matrix of the whole (CRAN-IDNC) graph. In other words, we need to check the \textbf{General Condition} \textbf{(GC)} of each pair of the total number of vertices in CRAN-IDNC graph $\mathcal{O}(\Perm{\mathcal{C}}{B}Z)$ to determine whether they should be connected with an edge, thus needing $\mathcal{O}(\Perm{\mathcal{C}}{B}Z)^2$. The complexity of the maximum-weight clique algorithm can be decomposed in sum weights and vertex search computations as follows. Each iteration needs $\mathcal{O}(\Perm{\mathcal{C}}{B}Z)$ operations for weight calculations of its maximum-weight clique. Note that each maximum-weight clique has at most $Z$ vertices as each subgraph can contribute at most with one vertex per transmission. Indeed, the complexity of the algorithm for identifying the maximum-weight clique and their sum weights requires at most $\mathcal{O}(\Perm{\mathcal{C}}{B}Z+ Z)$. Given the above configurations on the complexities of the different algorithm components, the total complexity of Algorithm \ref{Alg1} for each transmission is $\mathcal{O}(\Perm{\mathcal{C}}{B}Zc(A)+(\Perm{\mathcal{C}}{B}Z)^2+\Perm{\mathcal{C}}{B}Z+ Z)$ operations. Therefore, the total computing complexity of Algorithm \ref{Alg1} is $\mathcal{O}(\Perm{\mathcal{C}}{B}Z(c(A)+\Perm{\mathcal{C}}{B}Z+1)+ Z)$.
\end{proof}

It can be noted that the optimal joint solution needs a higher computational complexity. Thus, \sref{IOC} proposes the decoupling approach to efficiently approximate the joint solution.

\section{Iterative Optimization for Coordinated Scheduling and Power Control}\label{IOC}

In this section, we present an iterative coordinated scheduling and power control policy to solve the optimization problem \eref{eq4a}. In other words, for a fixed PL, the coordinated scheduling problem is addressed, and for a fixed scheduling, the power control problem is solved. The solution requires that the cloud collects and processes the coordinated scheduling and power control levels. The iteration between these two problems is carried out upon convergence.

To solve \eref{eq4a}, notice first that for a given feasible power levels PLs, the joint optimization problem \eref{eq4a} can be simplified to a coordinated scheduling problem only, and can be expressed as follows: 
\begin{subequations}
\begin{align}
&\max\sum _{b\in \mathcal{\mathcal{B}}} \sum _{z \in \mathcal{\mathcal{Z}}} \sum _{u\in\mathcal{\tau}_{bz}(\kappa_{bz})}X_{bz}^u R_{bz} \label{eq8a} \\
&{\rm s.t.\ } Y_{b}^u=\mathcal \min\left( \sum _{z}X_{bz}^u,1 \right),~\forall~ (b,u)\in \mathcal{B} \times \mathcal{U}, \label{eq8b} \\
& \sum _{b}Y_{b}^u\leqslant1,~\forall~ u\in \mathcal{U}, \label{eq8c} \\
& \tau_{bz}(\kappa_{bz}) = \left\{u \in \mathcal{U}\ \bigg|\kappa_{bz} \cap \mathcal{W}_u = 1 \mbox{ \& } R_{bz}\leqslant R^u_{bz}\right\}, \label{eq8d}
\\& X_{bz}^u , Y_{b}^u\in \{0,1\},{\kappa}_{bz}\in \mathcal{P}(\mathcal{F}), (u,b,z)\in \mathcal{U}\times \mathcal{B}\times\mathcal{Z}.\label{eq8f}
\end{align}
\end{subequations}

The optimization is carried over the variables $X_{bz}^u$, $Y_{b}^u$, $\kappa_{bz}$ and $R_{bz}$. The variables $Y_{b}^u$, $X_{bz}^u$, and $R_{bz}$ is continues optimization parameter. Constraints \eref{eq8b} and \eref{eq8c} translate the system condition \textbf{CC1}, i.e., each user must connect to at most RRH but possibly to many RRBs within RRH frame. In order to efficiently solve \eref{eq8a}, this paper uses graph theory techniques to map the feasible points to maximum cliques in a coordinated scheduling graph as explained in section \ref{COOS}. 

On the other hand, for any given user-RRBs/RRHs coordinated scheduling, the joint problem \eref{eq4a} can be considered as a power allocation step and reduces on a per-RRB basis. For each RRB $z$, the optimization problem can be written as:
\begin{align} \label{eq9}
&\max\sum _{b\in \mathcal{B}} \sum _{u\in {\mathcal{\tau}_{bz}(\kappa_{bz})}} \log_{2}(1+\text{SINR}^{u}_{bz}(\textbf{P}))
\nonumber \\
&{\rm s.t.\ } 0\leqslant P_{bz}\leqslant P^{\max}_{bz},~\forall~ b \in \mathcal{B}, 
\end{align}
where the optimization is over the set of powers $P_{bz}$, $\forall ~b\in\mathcal{B}$, where $\mathcal{B}$ is the sets of RRHs that have associations in the fixed schedule, and $u$ is the index of user that belongs to the set of targeted users $\mathcal{\tau}_{bz}(\kappa_{bz})$ in the fixed schedule.

To solve \eref{eq9}, the corresponding optimal power levels PLs must satisfy the Karush-Kuhn-Tucker (KKT) condition as explained in \sref{PAF}. Therefore, this section aims to propose an iterative solution to solve \eref{eq4a} approximately.

\subsection{Coordinated Scheduling For Fixed Power Levels}\label{COOS}
As stated earlier, for fixed power allocation, the joint coordinated scheduling and power allocation optimization problem \eref{eq4a} is reduced to a coordinated scheduling problem \eref{eq8a} only. Thus, in this subsection, we solve \eref{eq8a} using graph theory techniques by constructing the coordinated scheduling graph, in which each vertex represents the possible association of RRHs, users, files and achievable capacities of each RRB, and then reformulates the problem. Moreover, we use an efficient solution to solve \eref{eq8a}.

Let $\mathcal{A}$ be the set of all possible associations between RRHs, RRBs, users, files and the achievable capacity, i.e., $\mathcal{A}=\mathcal{B}\times\mathcal{Z}\times \mathcal{U}\times \mathcal{F}\times\mathcal{R}$. Let the coordinated scheduling graph be denoted by $\mathcal{G}(\mathcal{V},\mathcal{E})$ wherein $\mathcal{V}$ and $\mathcal{E}$ refer to the set of vertices and edges of this graph, respectively. 

This graph is constructed by generating a vertex $v$ for each possible association $ s \in \mathcal{A}$. In the same RRB $z$ and RRH $b$, two vertices $v \in \mathcal{V}$ associated with $s \in \mathcal{A}$ and $v^\prime \in\mathcal{V}$ associated with $s^\prime \in\mathcal{A}$ are connected by an edge if one of IDNC conditions (\textbf{C1} or \textbf{C2}) in \sref{RRB} and $\varphi_r(s)=\varphi_r(s^\prime)$ are true. This satisfaction ensures that all users represented by the associations have the same transmission rate, and receive always decodable transmission. 

Similarly, Two different vertices belonging to two different RRHs/RRBs are then set adjacent if their combination results in a feasible schedule, i.e., it satisfies the system constraint \textbf{CC1}. Let vertex $v\in \mathcal{G}$ be corresponding to the association $s \in \mathcal{A}$ and vertex $v^\prime\in \mathcal{G}$ corresponding to the association $s^\prime \in \mathcal{A}$. The vertices $v$, $v^\prime$ are adjacent if one of the \textbf{General Conditions (GC)} is satisfied.
\begin{itemize}
\item {\textbf{GC1}:} $\varphi_u(s)=\varphi_u(s')$ and $\varphi_b(s)=\varphi_b(s')$, $\forall~(s,s')\in  \mathcal{A}$. This condition translates the fact that the same user can be served with multiple RRBs within the same RRH. 
\item {\textbf{GC2}:}  $\left(\varphi_b(s) =\ \varphi_b(s^\prime)~\mbox{and}~\varphi_f(s)= \varphi_f(s')\right)$ OR $\left(\varphi_b(s) =\ \varphi_b(s^\prime)~\mbox{and}~\varphi_f(s)\in \mathcal{H}_{\varphi_u(s')}~\mbox{and}~ \varphi_f(s^\prime)\in \right.$ $\left.\mathcal{H}_{\varphi_u(s)}\right)$. This condition guarantees that the encoded combinations of the same users can be served by multiple RRBs within the same RRH.
\item {\textbf{GC3}:} $\varphi_u(s) \neq\ \varphi_u(s^\prime)$ and $\left(\varphi_b(s), \varphi_z(s)\right) \neq\ \left(\varphi_b(s^\prime),\varphi_z(s')\right)$. This condition completes the adjacencies in the graph for any two vertices not opposing the CC1 constraint for any two different users. 
\end{itemize} 

\fref{fig5} shows an example of the coordinated scheduling graph for a simple network consisting of $3$ users, $3$ files, $2$ RRHs, and $1$ RRB in each RRH frame. In this example, each vertex is labeled $bzufr$, where $b$, $z$, $u$, $f$ and $r$ represent the indices of RRHs, RRBs, users, files and achievable capacities, respectively. The Dashed and solid lines in Fig. \ref{fig5} represent the edges generated by the aforementioned local conditions and general conditions \textbf{(C1)}, \textbf{(LC1)} and \textbf{GC}, respectively, and the potential cliques in this example represented in the graph by solid lines are: $\{11111, 21221\}$,$\{21332, 11111\}$, $\{21221, 11332\}$, $\{21111, 21221, 21331\}$, $\{11111, 11221, 11331\}$, $\{21112, 11222, 11332\}$, and $\{11223, 21112, 21332\}$ achieving a total throughputs of 2, 3, 3, 3, 3, 6, and 7 bits/sec, respectively. Clearly, the red colored vertices represent the last maximal clique, and should be the one selected as it maximizes the throughput for this scheduling frame.

Given the above coordinated scheduling graph' construction
$\mathcal{G}(\mathcal{V},\mathcal{E})$, it can be established that any maximal clique in the graph satisfies that all users have vertices in the selected clique receive/overhear an instantly decodable transmission (satisfy XOR-IDNC and transmission rate conditions), and then can decode a new file.
\begin{figure}[t]
\centering
\includegraphics[width=\linewidth]{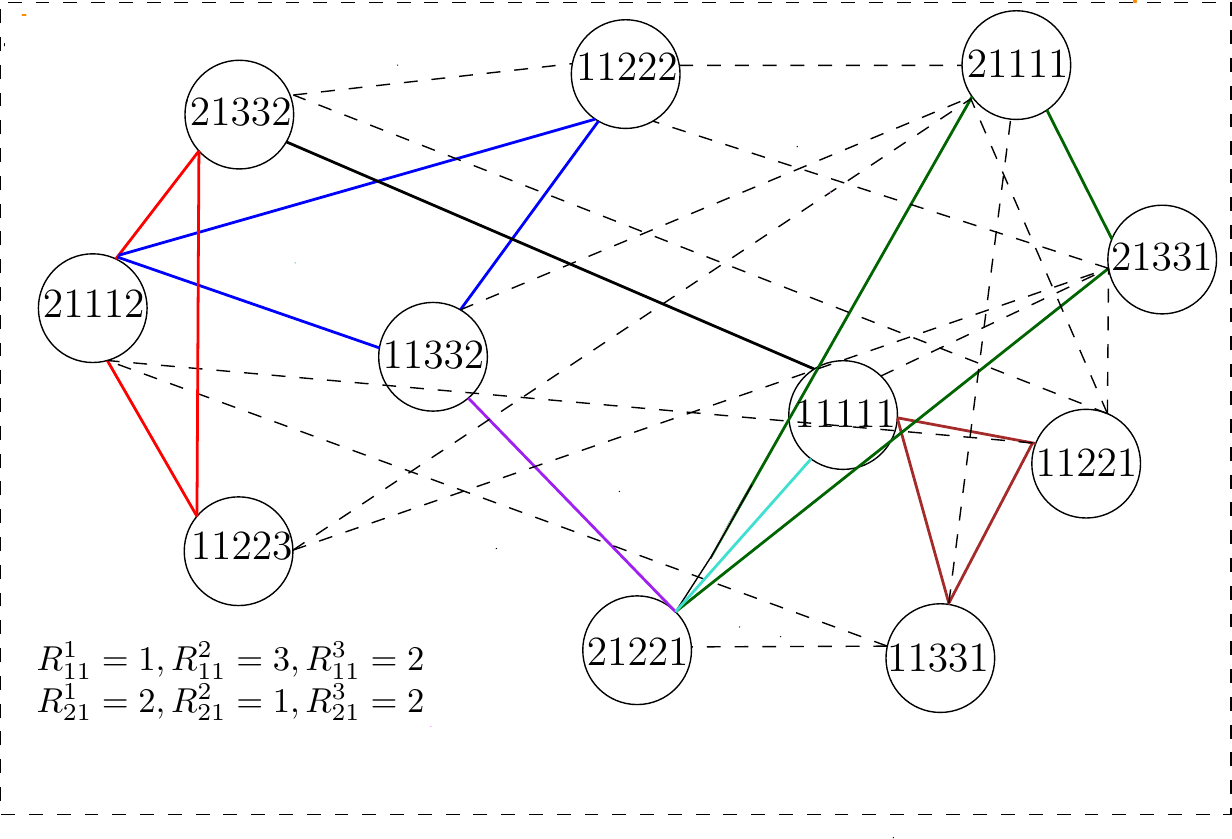}
\caption{The Coordinated Scheduling graph of the network presented in \fref{fig3} using Algorithm \ref{Alg2}.}
\label{fig5}
\end{figure}

\begin{algorithm}[t]
\caption{Coordinated Scheduling Algorithm}
\label{Alg2}
\begin{algorithmic}[]
\REQUIRE $\mathcal{U}, \mathcal{F}, \mathcal{B}, \mathcal{H}_u, \mathcal{W}_u, P_{bz},\mathcal{R}$ and $h^u_{bz}$,\\ $(b,z,u,f)\in\mathcal{B}\times\mathcal{Z}\times\mathcal{U}\times\mathcal{F}$.
\STATE Initialization: maximum-weight clique $\mathbf{M}\ = \phi$.
\STATE Construct $\mathcal{G}$ using \eref{COOS}.
\STATE Solve maximum-weight clique problem over $\mathcal{G}$ as follows:
\WHILE { $\mathcal{G} \neq\ \phi$} 
\STATE $\forall~$ vertex $v \in \mathcal{G}$: calculate $w(v)$ using \eref{eq10a}. 
\STATE Select $ v^\ast = \max_{v\in\mathcal{G}} \{{w(v)}\}$ 
\STATE Set $\mathbf{M}$ = $\mathbf{M}\ \cup v^\ast$
\STATE Set $\mathcal{G}=\mathcal{G}(v^\ast)$
\STATE Continue only with the vertices adjacent to $v^\ast$
\ENDWHILE
\STATE Output $\mathbf{M}$
\end{algorithmic}
\end{algorithm}

The following theorem characterizes the solution to the coordinated scheduling problem for fixed power allocation.
\begin{theorem}
The coordinated scheduling problem in \eref{eq8a} is equivalent to a maximum-weight clique problem over the coordinated scheduling graph, wherein the weight of a vertex $v \in \mathcal{V}$ corresponding to the association $s=(b,z,u,f,r) \in \mathcal{A}$ is given by
\begin{align} \label{eq10a}
w(v) = r.
\end{align}
\end{theorem}
\begin{proof}
This theorem can be proved by demonstrating the following facts. The first fact establishes a one-to-one mapping between the feasible schedules and the cliques in the coordinated scheduling graph. Afterward, the weight of each vertex is set to be the contribution of the corresponding user to the network. Therefore, the maximum weight clique is a feasible solution with the maximum received-throughput. In other words, the maximum-weight clique is the solution to \eref{eq8a}. The complete proof can be found in Appendix A in \cite{1}.
\end{proof}

In \cite{16}, it is shown that the maximum-weight clique problem is NP-hard. However, it is established that it can be optimally solved with reduced complexity as compared the $\mathcal{O}(|\mathcal{V}|^2.2^{|\mathcal{V}|})$ naive exhaustive search methods, e.g., the optimal algorithms in \cite{2155446,17,16513519}. The maximum-weight clique problem can be solved in linear time with respect to its size using the simple heuristic proposed shown in Algorithm \ref{Alg2}. 

\subsection{Power Allocation For Fixed Schedule} \label{PAF}
The power allocation step assumes a fixed user-RRBs/RRHs schedule and finds the optimal power levels PL of each RRB in \eref{eq9}. It can be easily noted that \eref{eq9} is a well known non-convex optimization problem, and its global solution is not feasible. Thus, the remaining of this section focuses on the numerical solution to achieve at least a local optimum solution. Therefore, our focus in this section is to solve the optimization problem \eref{eq9} using the Karush-Kuhn-Tucker (KKT) iteration approach. In particular, for a given user-RRB/RRH selections, the corresponding optimal set of powers must satisfy the first derivative, i.e., KKT condition. Therefore, the objective function of the problem \eref{eq9} which is optimized over the set of power on a RRB-by-RRB basis, can be expressed as:

\begin{align} \label{eq10} \nonumber
& R(P^z_1,P^z_2,...,P^z_B)=\\& \sum _{b\in \mathcal{B}} \sum _{u\in {\mathcal{\tau}_{bz}(\kappa_{bz})}} \nonumber \log_{2}\bigg(1+\cfrac{P_{bz} |h^{u}_{bz}|^{2}}{\sigma^2 +\sum\limits_{b'\neq b}P_{b'z}| h^{u}_{b'z}|^{2}}\bigg)
\\
&{\rm s.t.\ } 0\leqslant P_{bz}\leqslant P^{\max}_{bz},~\forall~ b\in \mathcal{B}. 
\end{align}

 We start by taking the first derivative of the objective function \eref{eq10} with respect to $P_{bz}$: 
\begin{align}
&\frac{\partial R}{\partial P_{bz}}=\frac{\partial }{\partial P_{bz}}\sum _{u\in {\mathcal{\tau}_{bz}(\kappa_{bz})}} \log_{2}\bigg(1+\cfrac{P_{bz} |h^{u}_{bz}|^{2}}{\sigma^2 +\sum\limits_{b'\neq b}P_{b'z}| \nonumber h^{u}_{b'z}|^{2}}\bigg)\\&+\frac{\partial }{\partial P_{bz}} \sum\limits_{b'\neq b}\sum _{u'\in {\mathcal{\tau}_{b'z}(\kappa_{b'z})}}\nonumber 
{\log_{2} \bigg(1+\cfrac{P_{b'z} |h^{u'}_{b'z}|^{2}}{\sigma^2 +\sum\limits_{bb\neq b'}P_{bbz}| h^{u'}_{bbz}|^{2}}\bigg)}\\&=\frac{1}{P_{bz}}\sum \nonumber _{u\in{\mathcal{\tau}_{bz}(\kappa_{bz})}} \left( \frac{\text{SINR}^{u}_{bz}}{1+\text{SINR}^{u}_{bz}}\right)\\&-\sum\limits_{b'\neq b}\sum _{u'\in{\mathcal{\tau}_{b'z}(\kappa_{b'z})}}\frac{|h^{u'}_{bz}|^{2}}{P_{b'z} |h^{u'}_{b'z}|^{2}} \left( \frac{(\text{SINR}^{u'}_{b'z})^2}{1+\text{SINR}^{u'}_{b'z}}\right)
\end{align}

where,
\begin{align}
&\ \text{SINR}^{u'}_{b'z}= \bigg(1+\cfrac{P_{b'z} |h^{u'}_{b'z}|^{2}}{\sigma^2 +\sum\limits_{bb\neq b'}P_{bbz}| h^{u'}_{bbz}|^{2}}\bigg)\label{eq11} 
\end{align}
and $u\in {\mathcal{\tau}_{bz}(\kappa_{bz})}$ and $u'\in {\mathcal{\tau}_{b'z}(\kappa_{b'z})}$ are the scheduled users of the $b$-th RRH, and the $b'$-th RRH at the $z$th RRB for $\forall ~b$ and $~b'$ $\in\mathcal{B}$, respectively. By letting the above gradient equal to zero and by manipulating the optimality condition, one can obtain this manipulation for optimizing the power: 
\begin{align} 
&\frac{1}{P_{bz}}\sum \nonumber _{u\in{\mathcal{\tau}_{bz}(\kappa_{bz})}} \left( \frac{\text{SINR}^{u}_{bz}}{1+\text{SINR}^{u}_{bz}}\right)\\&=\sum\limits_{b'\neq b}\sum _{u'\in{\mathcal{\tau}_{b'z}(\kappa_{b'z})}}\frac{|h^{u'}_{bz}|^{2}}{P_{b'z} |h^{u'}_{b'z}|^{2}} \left( \frac{(\text{SINR}^{u'}_{b'z})^2}{1+\text{SINR}^{u'}_{b'z}}\right)
\label{eq13}
\end{align} 
Therefore, 
\begin{align}
&P_{bz}=&\frac{\sum\limits_{u\in{\mathcal{\tau}_{bz}(\kappa_{bz})}}\left(\frac{\text{SINR}^{u}_{bz}}{1+\text{SINR}^{u}_{bz}}\right)}{\sum\limits_{b'\neq b}\sum\limits_{u'\in{\mathcal{\tau}_{b'z}(\kappa_{b'z})}}\frac{|h^{u'}_{bz}|^{2}}{P_{b'z} |h^{u'}_{b'z}|^{2}} \left( \frac{(\text{SINR}^{u'}_{b'z})^2}{1+\text{SINR}^{u'}_{b'z}}\right)}
\label{eq14}
\end{align}
The KKT condition \eref{eq14} is essentially a water-filling condition if the dominator term is fixed. In this case, \eref{eq14} gives the following power update equation, i.e., “KKT method”: (for more details see \cite{4}, \cite{W.} and references therein).
\begin{align}
&P_{bz,new}=&\left[ \frac{\sum\limits_{u\in{\mathcal{\tau}_{bz}(\kappa_{bz})}} \left(\frac{\text{SINR}^{u}_{bz}}{1+\text{SINR}^{u}_{bz}}\right)}{\sum\limits_{b'\neq b}t_{b'z}}
\right]^{P_{bz}^{\max}}_{0} \label{eq15} 
\end{align}
where, \begin{align}
&\ t_{b'z}= \sum _{u'\in {\mathcal{\tau}_{b'z}(\kappa_{b'z})}}\frac{|h^{u'}_{bz}|^{2}}{P_{b'z} |h^{u'}_{b'z}|^{2}} \left( \frac{(\text{SINR}^{u'}_{b'z})^2}{1+\text{SINR}^{u'}_{b'z}}\right) \label{eq16} 
\end{align}

It is important to note that, the nominator in the right hand side of \eref{eq15} represents the effect power of $z$-th RRB in $b$-th RRH on all corresponding RRHs in the schedule, i.e., it is the derivative of the $b'$-th RRHs terms with respect to the $z$-th RRB power in the $b$th RRH. In other words, it summarizes the interfering effect of $P_{bz}$ on the $b'$th RRH. Moreover, it depends on the transmit power, SINR and the ratio of the direct and the interfering channel gains. The dominator in the right hand side of \eref{eq15} shows the effect of the combined noise and interference in RRB $z$ of the RRH $b$.
\subsection{Proposed Iterative Algorithm}
To summarize the iterative solution in this section, for fixed power, we solve the coordinated scheduling problem as explained in \sref{COOS} and Algorithm \ref{Alg2}. Afterwards, for a fixed schedule, the power allocation problem is solved as in \sref{PAF} and updated based on \eref{eq15}. The Process of iterating between coordinated scheduling and power allocation steps continue until convergence. Neither the coordinated scheduling step nor the power allocation step is nondecreasing in the optimization objective. Therefore, the iterations converge.
\section{Simulation Results} \label{SR}
This section shows the performance of the proposed solutions in the downlink of a C-RAN described in \sref{SMMM}. The network model and the physical layer model are implemented in MATLAB. The total number of RRHs is fixed to $3$. Users are uniformly distributed within the cell. To study the performance of the proposed algorithms in various scenarios, number of users, number of RRBs per each RRH frames, maximum power $P^{\max}$, cell size $C$, and distribution of the side information vary so as to study multiple scenarios. The additional simulation parameters are summarized in \tref{table_example}. The performance of the proposed solution is compared to the state-of-the-art coded and uncoded methods. In particular, the implemented schemes in this paper are: 
\begin{itemize}
\item Classical IDNC (rate-unaware scheme): This scheme
jointly optimizes the selection of an XOR file combination
for each RRB in each RRH without considering the achievable capacities of users. After the file selection process, the CRAN's physical-layer employs the minimum achievable capacity of all users targeted by each RRB as its transmitting rate.
\item RLNC (rate-greedy scheme): In this scheme, each user is
associated with a single RRB to which it has the maximum
capacity. If more than one user is associated with
the same RRB, random linear network coding (RLNC)
is employed. The encoding is done irrespectively of the side information. Indeed, as stated earlier, RLNC mixes
all files with different random coefficients. The selected
transmission rate in each RRB is the capacity of users
having the minimum achievable capacity in that RRB.
\item Maximum power transmission scheme: In this scheme, the transmission power of the RRBs is set to the maximum power levels $P^{\max}$. 
\item Joint coordinated scheduling and power control (uncoded scheme): In this scheme, only one user is served in each RRB, each user can be assigned to more than one RRBs from the same RRH. The user-RRB association is proposed in \cite{18a} so as to maximize the sum rate of the CRAN. 
\item CRAN-IDNC (rate-aware scheme): This scheme is described in Section \ref{JSP}.
\item Iterative coordinated scheduling and power allocation (rate-aware scheme): This scheme is described in Section \ref{IOC}.
\end{itemize}
\begin{table}[!t]
\renewcommand{\arraystretch}{0.9}
\caption{Simulation Parameters}
\label{table_example}
\centering
\begin{tabular}{|c|c|}
\hline
Cellular Layout & Hexagonal \\
%\hline
%Cell Diameter & 500 meters \\
\hline
Channel Model & SUI-Terrain type B \\
\hline
Channel Estimation & Perfect \\
\hline
Background Noise Power &-168.60 dBm/Hz\\
\hline
Shadowing Variance & 0dB \\
\hline
Bandwidth &10 MHz\\
\hline
\end{tabular}
\end {table}

\begin{figure}[t]
\centering
\includegraphics[width=\linewidth]{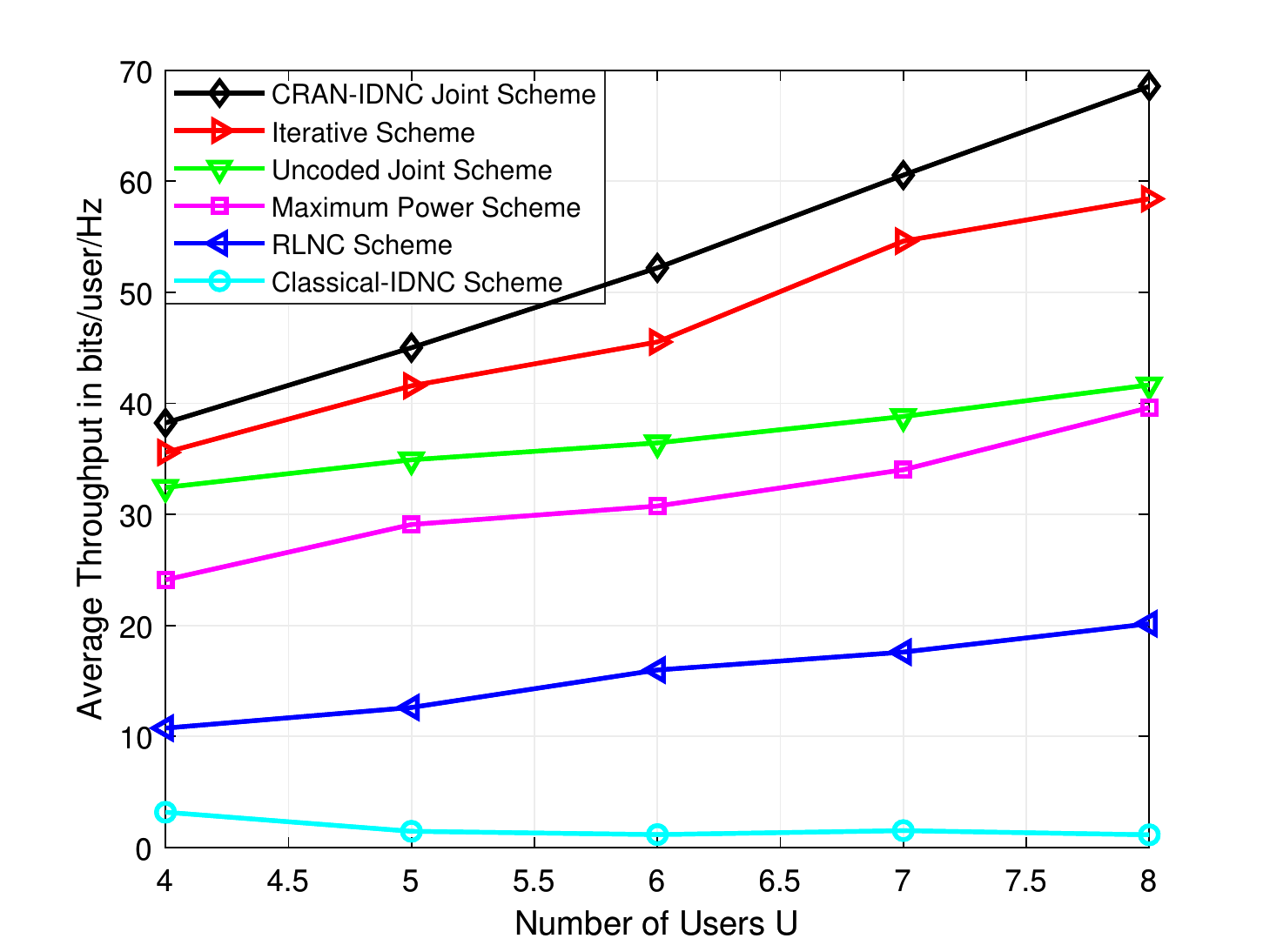}
\caption{Average Throughput in bits/user/Hz. Vs the number of users $U$. Number of RRB is $2$ each send with maximum power $P^{\max} = -42.60$ dBm, and cell size $C = 500$m.}
\label{fig6}
\end{figure}

\begin{figure}[t]
\centering
\includegraphics[width=\linewidth]{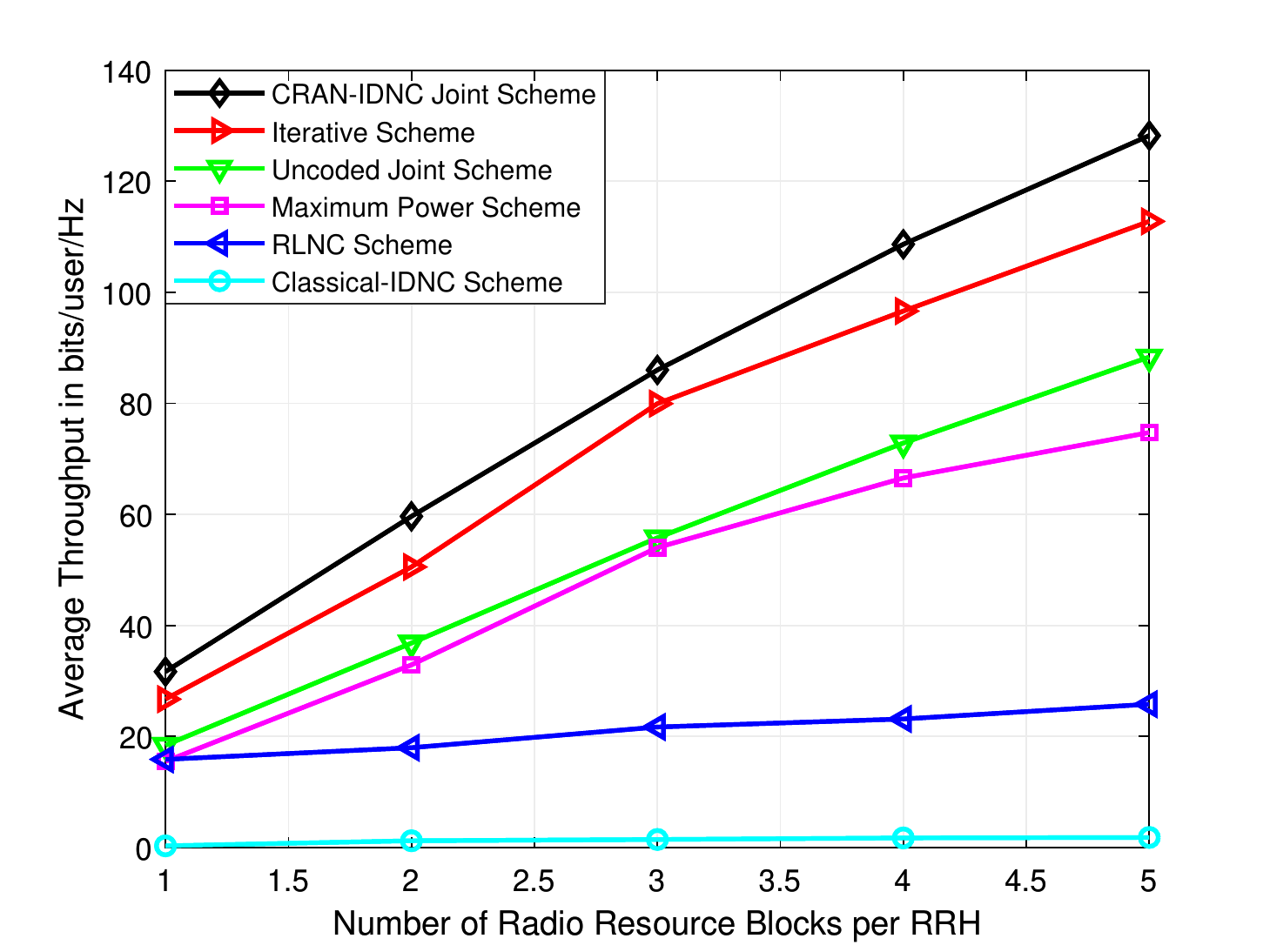}
\caption{Average Throughput in bits/user/Hz. Vs the number of Radio Resource Blocks $Z$. Number of users is $7$, a maximum transmit power $P^{\max} = -42.60$ dBm, and cell size $C = 500$m.}
\label{fig7}
\end{figure}
\fref{fig6} depicts the average throughput in bits/user/Hz achieved by our proposed algorithms and the aforementioned schemes for different numbers of users $U$, given a CRAN composed of $2$ RRBs per RRH's frame, a file size $N=1$Mb, a maximum transmit power $P^{\max} = -42.60$ dBm, and a cell size $C = 500$m. From the figure, we note that our proposed CRAN-IDNC scheme outperforms the iterative and the other schemes. In particular, the joint optimal uncoded scheme only focuses on the high achievable rates at the expense of transmitting at most one file to a single user from each RRB in all RRHs, i.e., a maximum number of targeted users is $Z_{tot}$. On the other hand, the maximum power and the RLNC schemes serve a good number of users in each transmission but sacrifice the optimality of the power and the rate. One can also notice that the gap between our proposed scheme and the other schemes increases as the number of users increases. The gain is due to the fact that the proposed scheme benefits from the increasing number of users by mixing the flows of more and more users to the same RRB plus the role of the CRAN-IDNC scheme as an interference mitigating technique that increases with the increase in the number of users. 

\fref{fig7} shows the average throughput in bits/user/Hz versus the numbers of RRBs $Z$ for a CRAN composed of  $7$ users, a file size $N=1$Mb, a maximum transmit power $P^{\max} = -42.60$ dBm, and cell size $C = 500$m. Again, the figure shows that our proposed CRAN-IDNC scheme outperforms all other schemes. The gap in performance increases as the number of RRBs per frame grows. It can also be easily seen from the figure that the performances of both proposed schemes and the joint optimal uncoded scheme increase linearly with the increase in the number of RRBs with a fixed number of users. In fact, all schemes agree in serving the same user in different RRBs of the same RRH. Therefore, increases in the number of RRBs increases the total received throughput. It can be noted from the figure that as the number of RRBs increases our CRAN-IDNC scheme outperforms the iterative solution. This can be explained by the fact that, increasing the number of RRBs leads to more and more power control subgraphs. Thus the size of the search space becomes larger, which works in favour of our CRAN-IDNC algorithm.

\fref{fig8} plots the average throughput as a function of the file size $N$ in a CRAN system composed of $7$ users and $2$ RRBs per RRH's transmit frame, each RRB has a maximum transmit power $P^{\max} = -42.60$ dBm, and a cell size $C = 500$m. As the file's size increases, the performance of all schemes increases. The figure shows that all schemes increase linearly with the size of the file. This can be explained by the fact that, as the size of the file increases, more and more bits are received, thus increasing the average received throughput. 
\begin{figure}[t]
\centering
\includegraphics[width=\linewidth]{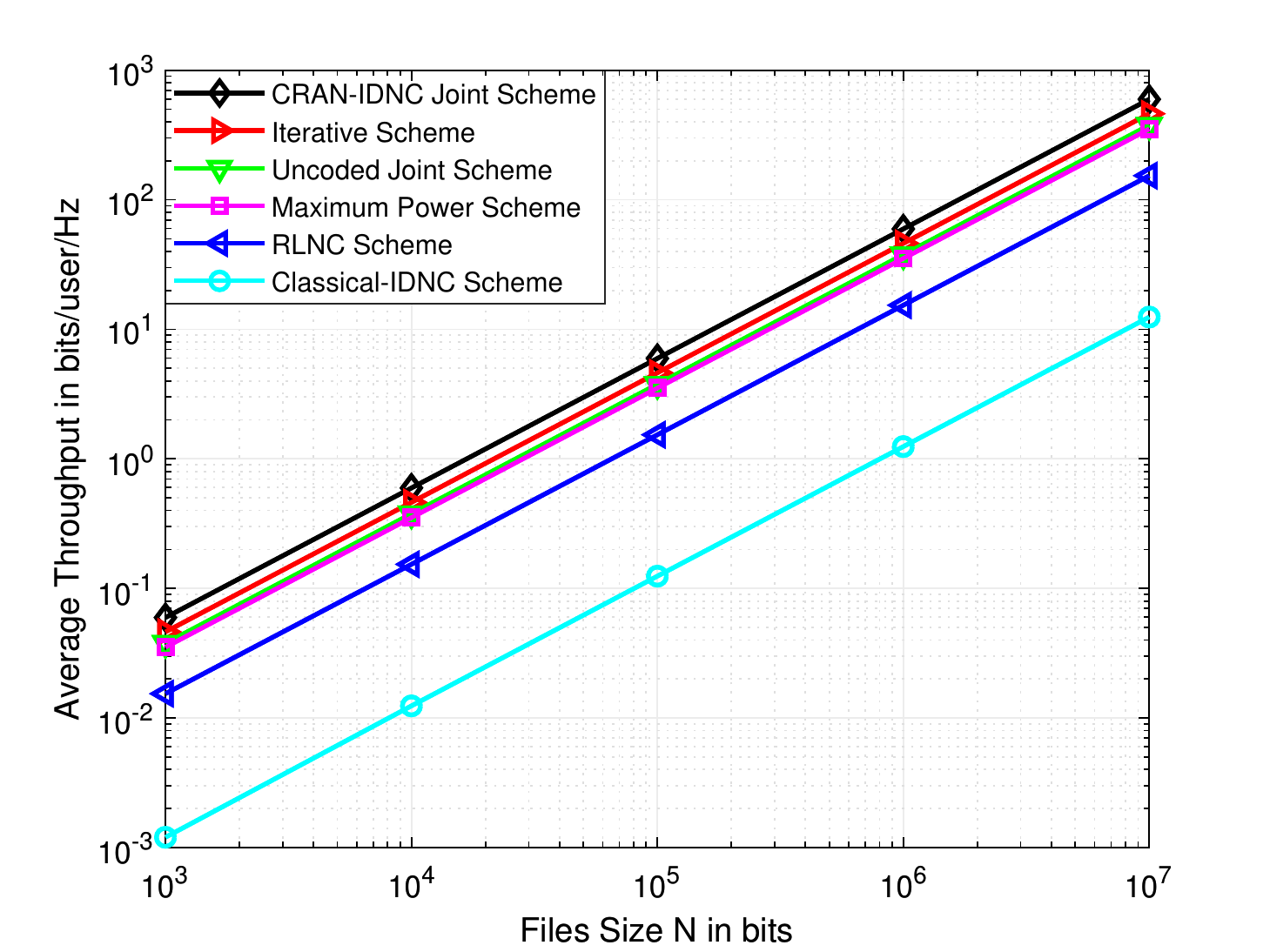}
\caption{Average Throughput in bits/user/Hz. Vs the file's size $N$ in bits. Number of users is $7$, number of RRBs is $2$ each has a maximum transmit power $P^{\max} = -42.60$ dBm, and cell size $C = 500$m.}
\label{fig8}
\end{figure}

\fref{fig9} plots the average throughput in bits/user/Hz versus
the maximum power $P^{\max}$, for a CRAN setting composed of $2$ RRBs in each RRH frame, $7$ users, a file size $N=1$Mb, and cell size $C = 500$m. The figure shows that our CRAN-IDNC optimization algorithm outperforms all other schemes, particularly for large maximum power. The increase in performance can be explained by the fact that as the maximum allowed power increases, the inter-RRHs interference increases, which works in favor of CRAN-IDNC algorithm as a method for interference reduction.

\fref{fig10} plots the average throughput in bits/user/Hz versus the cell size $C$, for a CRAN setting composed of $7$ users, $2$ RRBs in each RRH' frames each RRB has a maximum allowed power $P^{\max} = -26.98$ dBm, and a file size $N=1$Mb. The proposed CRAN-IDNC algorithm largely outperforms the iterative solution particularly for small cell network, i.e., high interference level. As the cell size increases, i.e., low interference level, the performance of CRAN-IDNC solution over the iterative one decreases and gains $10\%$ improvement. Finally, despite its great merits in reducing the completion time of a frame of files in many prior works, Classical IDNC exhibits a very poor performance from a physical layer perspective, thus voiding its upper layer gains. This clearly shows the importance of rate-awareness in IDNC code design to achieve gains on both the upper and physical layers, thus leading to a real reduction in file delivery times.

\begin{figure}[t]
\centering
\includegraphics[width=\linewidth]{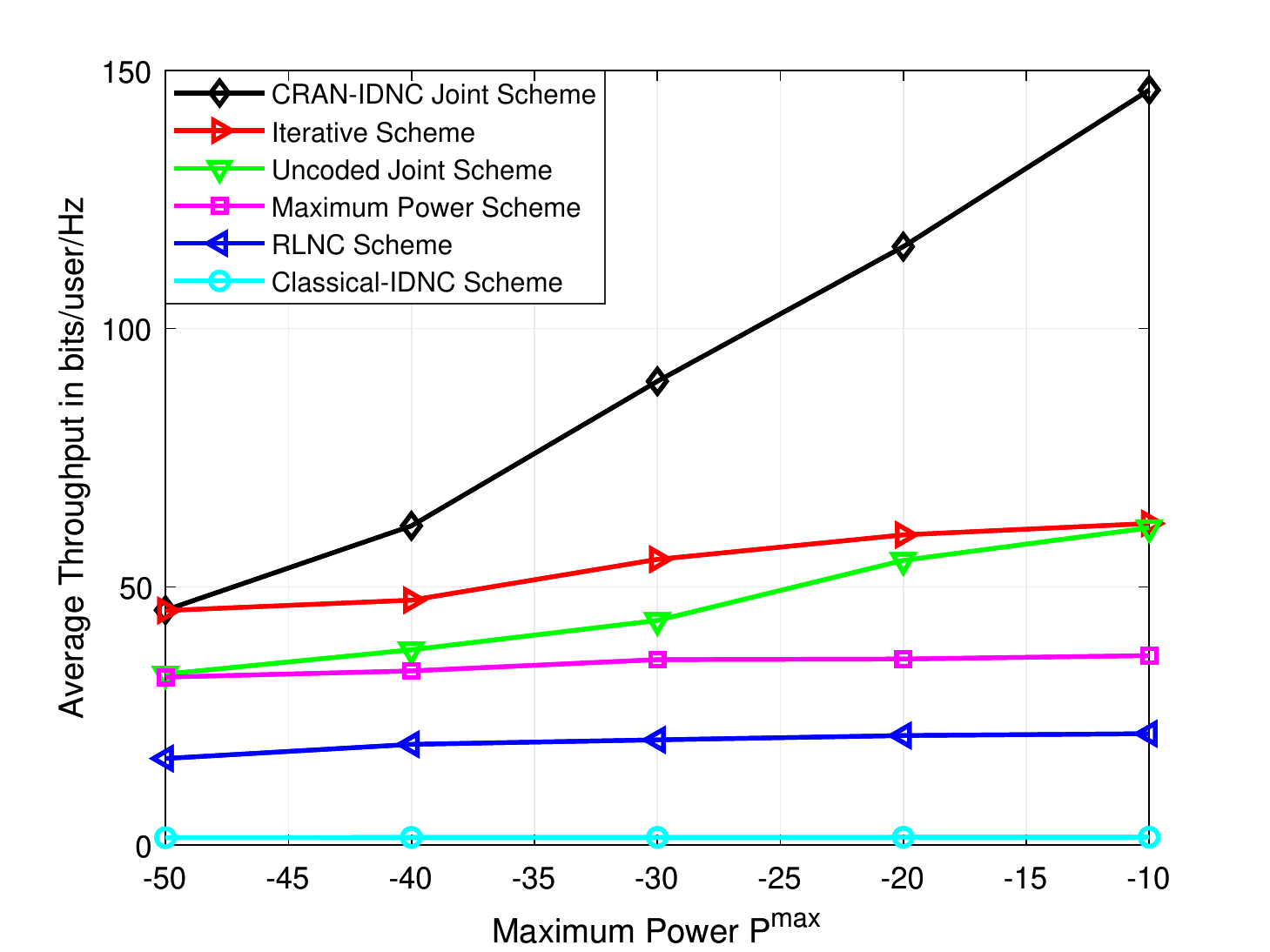}
\caption{Average Throughput in bits/user/Hz. Vs Maximum Power $P^{\max}$. Number of users is $7$, number of RRBs is $2$, file size is $1$Mb, and cell size $C = 500$m.}
\label{fig9}
\end{figure}

\begin{figure}[t]
\centering
\includegraphics[width=\linewidth]{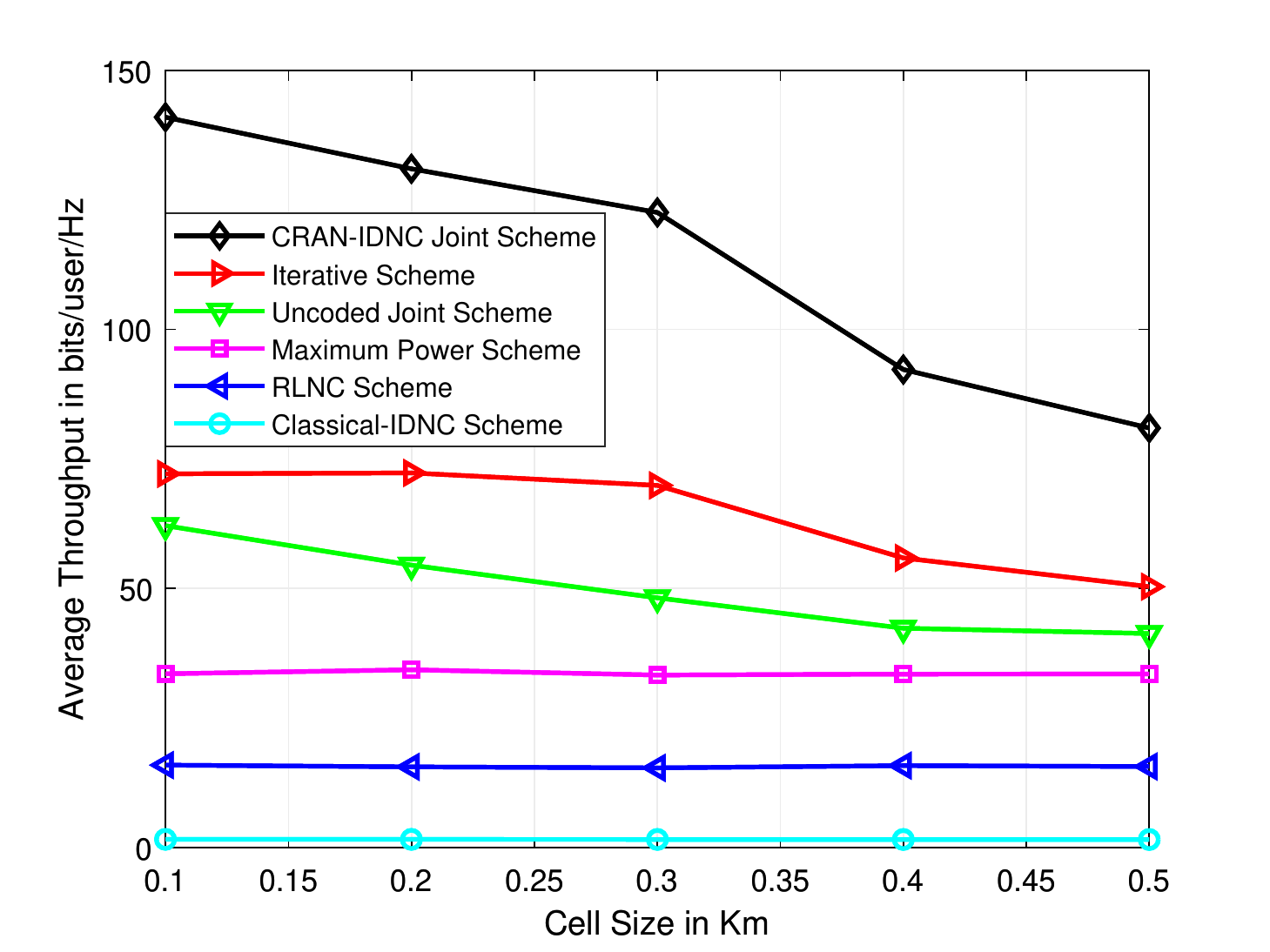}
\caption{Average Throughput in bits/user/Hz. Vs Cell Size $C$ in $K_m$. Number of users is $7$, number of RRBs is $2$ each has a maximum transmit power $P^{\max} = -26.98$ dBm, file size is $1$Mb.}
\label{fig10}
\end{figure}

\section{Conclusion} \label{CC}
Interference mitigation and resource blocks' efficient exploitation in Next Generation System 5G is an emerging topic of interest. This paper investigates the cross-layer optimization in cloud-enabled networks in order to solve the throughput maximization problem. Unlike previous studies that only considered the CRAN system from a physical layer perspective, we proposed to use the information available in the network to combine files using instantly decodable network coding. Therefore, the throughput maximization problem becomes the same as the problem of assigning users efficiently to the available resource blocks, choosing the file combination and the power allocations (PLs) of each under the constraint that a user can connect to at most a single remote radio head but to many resource blocks within it. A graph theoretical approach is proposed to solve the problem by designing the CRAN-IDNC graph formed by several power control subgraphs. By establishing a correspondence between the feasible solution to the problem and the cliques in the graph, the problem is shown to be equivalent to a maximum-weight clique which can be efficiently solved using state of the art methods. Simulation results show the performance of the proposed two solutions and reveal that they outperform uncoded and rate-unaware coding solutions.

%\bibliographystyle{IEEEtran}
%\bibliography{citations}

\end{document}